\documentclass[]{article}

\usepackage[utf8x]{inputenc}
\usepackage{graphicx}
\usepackage{amsmath}
\usepackage{amssymb}
\usepackage{braket}
\usepackage{hyperref}
\usepackage{cleveref}
\usepackage{geometry}
\usepackage{xcolor}
\usepackage{amsthm}
\usepackage[font=small,labelfont=bf]{caption}
\usepackage{booktabs}
\usepackage[table]{xcolor} 
\usepackage{subcaption}
\usepackage{quantikz}
\usepackage{soul}
\usepackage{enumitem}
\setlist{nosep}
\usepackage{float}
\usepackage{multicol}
\usepackage{algorithm}
\usepackage{dirtree}
\usepackage{tcolorbox}
\usepackage{fancyvrb}
\usepackage[frozencache,cachedir=.]{minted}

\tcbuselibrary{minted,breakable,xparse,skins}
\usepackage{tikz}
\usetikzlibrary{fit,arrows.meta,positioning,shapes.geometric}
\newcommand\undermat[2]{%
  \makebox[0pt][l]{$\smash{\underbrace{\phantom{%
    \begin{matrix}#2\end{matrix}}}_{\text{$#1$}}}$}#2}

\definecolor{bg}{gray}{0.95}
\DeclareTCBListing{mintedbox}{O{}m!O{}}{%
  breakable=true,
  listing engine=minted,
  listing only,
  minted language=#2,
  minted style=default,
  minted options={%
    linenos,
    gobble=0,
    breaklines=true,
    breakafter=,,
    fontsize=\small,
    numbersep=8pt,
    #1},
  boxsep=0pt,
  left skip=0pt,
  right skip=0pt,
  left=25pt,
  right=0pt,
  top=3pt,
  bottom=3pt,
  arc=5pt,
  leftrule=0pt,
  rightrule=0pt,
  bottomrule=2pt,
  toprule=2pt,
  colback=bg,
  colframe=orange!70,
  enhanced,
  overlay={%
    \begin{tcbclipinterior}
    \fill[orange!20!white] (frame.south west) rectangle ([xshift=20pt]frame.north west);
    \end{tcbclipinterior}},
  #3}

\usepackage[
	backend=biber,
	style=numeric,
    sorting=none,
	sortcites=true,
	url=true,
	eprint=true,
	giveninits=false,
    uniquename=true,
	minnames=1,
	maxnames=3
	]{biblatex}

\AtBeginBibliography{
  \hypersetup{
    urlcolor=black,
    hidelinks
  }
}
\AtEveryBibitem{
	\clearname{editor}%
	\clearfield{series}%
	\clearfield{isbn}%
	\clearfield{issn}%
    \clearfield{issue}
	\clearfield{day}%
	\clearfield{month}%
	\clearfield{place}%
	\clearlist{location}%
    \clearfield{doi}%
	\clearfield{eventtitle}%
	\ifentrytype{inproceedings}{%
		\clearlist{publisher}
	}{}%
	\ifentrytype{misc}
    {%
    }
    {%
		\clearfield{url}%
		\clearfield{urlyear}%
	}%
}
\title{QuickQudits: A Framework for Efficient Simulation of Noisy Qudit Clifford Circuits
via an Extended Stabilizer Tableau Formalism}

\makeatletter
\renewcommand\@date{{%
  \vspace{-\baselineskip}%
  \large\centering
  \begin{tabular}{@{}c@{}}
    Nina Brandl\textsuperscript{1} \\
    \normalsize nina.brandl@jku.at
  \end{tabular}%
  \quad
  \begin{tabular}{@{}c@{}}
    Mykyta Cherniak\textsuperscript{1} \\
    \normalsize mykyta.cherniak@jku.at
  \end{tabular}
  \quad
    \begin{tabular}{@{}c@{}}
    Johannes Kofler\textsuperscript{1} \\
    \normalsize johannes.kofler@jku.at
  \end{tabular}%
  \quad
  \begin{tabular}{@{}c@{}}
    Richard Kueng\textsuperscript{1} \\
    \normalsize richard.kueng@jku.at
  \end{tabular}%
  \bigskip

  \textsuperscript{1}Department for Quantum Information and Computation at Kepler (QUICK)\\ 
  Johannes Kepler University Linz, Austria

  \bigskip

  \text{24. March, 2026}
}}
\makeatother

\begin{document}

\maketitle
\newtheoremstyle{mythrm}%
    {4pt}{4pt}
    {}
    {}
    {\bfseries}
    {:}
    {0.5em}
    {}
\theoremstyle{mythrm}
\newtheorem{definition}{Definition}
\newtheorem{proposition}{Proposition}
\newtheorem{theorem}{Theorem}
\newtheorem{fact}{Fact}
\newtheoremstyle{myexample}%
    {4pt}{4pt}
    {}
    {}
    {\itshape}
    {:}
    {0.5em}
    {}
\theoremstyle{myexample}

\newcommand*\newtushtheorem[1]{%
  \AddToHook{env/#1/before}{\par\vspace{\baselineskip}\hrule width \hsize \kern 1mm \hrule width \hsize height 2pt}%
  \AddToHook{env/#1/after}{\hrule\vspace{\baselineskip}}%
  \newtheorem{#1}%
}
\newtushtheorem{example}{Example}

\newcommand*\newconttheorem[1]{%
  \AddToHook{env/#1/before}{\par\vspace{\baselineskip}\hrule width \hsize \kern 1mm \hrule width \hsize height 2pt}%
  \AddToHook{env/#1/after}{\hrule\vspace{\baselineskip}}%
  \newtheorem*{#1}%
}
\newconttheorem{excont}{\continuation}
\newcommand{\continuation}{??}
\newenvironment{continueexample}[1]
 {\renewcommand{\continuation}{\Cref{#1}}
 \par\vspace{\baselineskip}\hrule width \hsize \kern 1mm \hrule width \hsize height 2pt
 \excont[continued]}
 {\endexcont
 \hrule\vspace{\baselineskip}
 }

\newtheorem{corollary}{Corollary}

\begin{abstract}
We present a comprehensive and self-contained framework for the efficient classical simulation of Clifford circuits acting on $d$-dimensional qudits, including realistic Pauli/Weyl noise via stochastic simulation. Our approach uses the stabilizer tableau formalism for qudits of arbitrary dimension and tracks both stabilizer and destabilizer generators under Clifford updates. The classical simulation remains efficient with simple algebraic Clifford update rules over $\mathbb{Z}_d$. 
Computational basis measurements in prime dimensions are handled by a generalized Aaronson-Gottesman (CHP) procedure. 
In composite dimensions, $\mathbb{Z}_d$ is not a field and the standard tableau reduction fails, so we employ an exact Smith normal form decomposition to enable efficient sampling.
Noise is modelled as probabilistic mixtures of Weyl operators that act only on the tableau's phase column. For fast simulation of noisy circuits, we support Pauli frames, respectively generalized Weyl frames, and introduce a noise-pushing technique that allows all noise processes to be consolidated into a single phase update at the end of the circuit. Using this representation, circuit fidelity can be determined entirely by the single accumulated phase-shift parameter $\Delta \tau$, reducing fidelity estimation to a simple phase check per shot.
Our codebase supports tableau simulation and conventional state-vector and density-matrix backends for qudits, and also includes circuit and tableau visualisations. Additionally, we provide tests and Jupyter notebooks for validation and illustration. This framework forms the basis for a scalable, open-source strong+weak stabilizer simulator including noise and can be found publicly available at \url{https://github.com/QUICK-JKU/QuickQudits}.
\end{abstract}

\begin{figure}[t!]
    \centering
    \includegraphics[width=\linewidth]{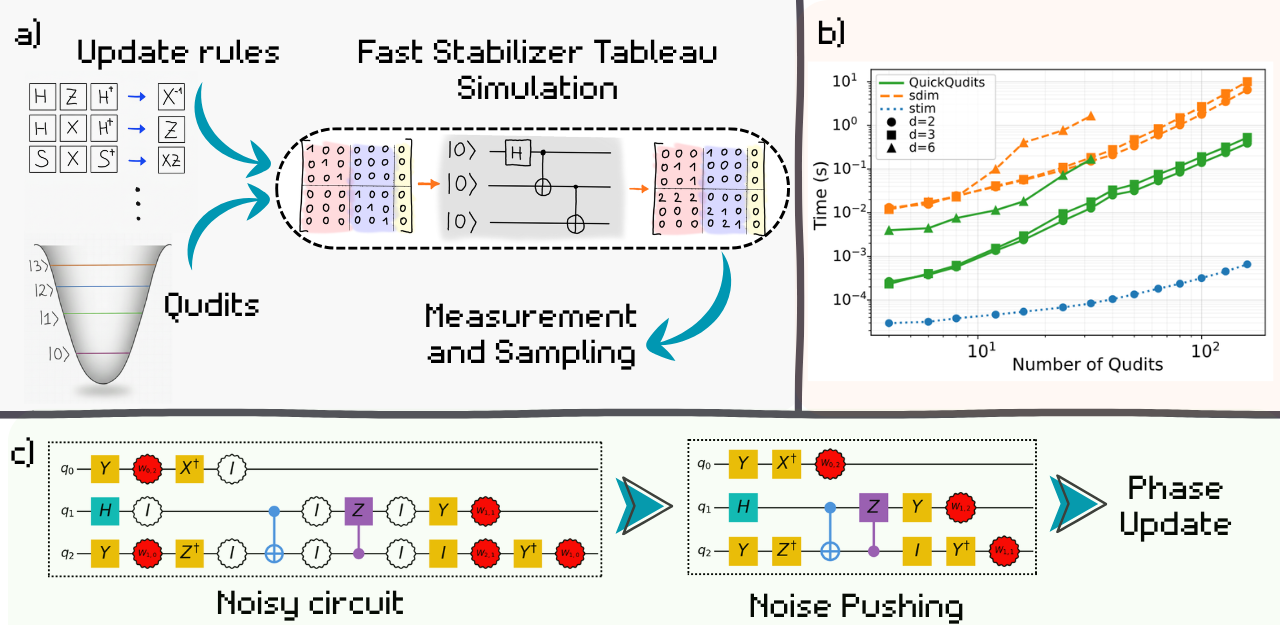}
    \caption{Overview of the qudit Clifford simulator and its main components. 
\textbf{a)} The simulator implements algebraic update rules for Clifford gates within the stabilizer tableau formalism, enabling efficient simulation of $n$-qudit circuits of local dimension $d$, including measurement, exact sampling, and visualization utilities. 
\textbf{b)} We benchmark the runtime against state-of-the-art stabilizer simulators such as \emph{Sdim} and \emph{Stim}, demonstrating competitive performance (see \Cref{sub:compare}). 
\textbf{c)} Noisy Clifford circuits are supported via Pauli frames and a noise-pushing technique that converts stochastic Weyl errors into explicit tableau phase updates for circuit-fidelity estimation.}
\label{fig:placeholder}
\end{figure}

\begin{multicols}{2}

\section{Introduction}
The simulation of quantum circuits on classical hardware remains one of the central topics in quantum information science. While arbitrary quantum computations require exponential classical resources, as far as we know, certain restricted classes can be simulated efficiently. Among them is the class of Clifford circuits, which are composed exclusively of Clifford gates, stabilizer-state inputs and computational basis measurements.

The Gottesman-Knill theorem \cite{gottesman1998heisenberg} shows that efficient classical simulations of Clifford circuits are possible for multi-qubit systems. However, many experimental platforms such as trapped ions \cite{ringbauer2022universal}, superconducting circuits \cite{fischer2023universal}, and photonic systems naturally exhibit an underlying Hilbert space that is higher-dimensional. The study of such qudits ($d$-level systems) is not new, but recently gained much traction in the community, because qudits offer novel approaches to tackle some of the most important challenges in the field. The qudit algebraic structure, for example, provides insights into higher-level quantum error-correcting codes \cite{ketkar2006nonbinary} and dimension-dependent noise resilience. It also proposes new challenges for efficient simulation. In this paper, we first explain the necessary theoretical background and then present our fully self-contained and open-source framework for the efficient and scalable simulation of noisy qudit Clifford circuits. \\
During the preparation of this paper, another qudit stabilizer tableau simulator called \emph{Sdim} has been published \cite{kabir2025sdim}. Our simulator toolbox has been developed independently and without prior knowledge of this work and serves as an extension to earlier results by one of the present authors \cite{brandl2024efficient}. Our simulator supports all local dimensions, tracks stabilizer phases via both Pauli-frame and phase-update representations, provides reduced post-measurement tableaus along with custom circuit visualizations and a strong matrix-vector simulator. In contrast, existing tools like \textit{Sdim} or \textit{Stim} \cite{gidney2021stim} are either limited in the dimension support, lack full phase tracking, or do not offer explicit post-measurement tableau updates or unified noise semantics. A more comprehensive comparison is given in \Cref{sub:compare}.

A strong simulator computes exact probabilities for all measurement outcomes, while a weak simulator samples outcomes according to the correct distribution without computing probabilities. Our framework supports weak simulation efficiently for all dimensions, while strong simulation with full post-measurement tableaus is currently limited to prime $d$.
In order to build a practical quantum simulator, we need to consider several key concepts: 
\begin{description}[leftmargin=10pt]
    \item[Representation:] We need a uniform and well defined algebraic representation \cite{hostens2005stabilizer} for storing stabilizer and destabilizer information and to systematically encode quantum states. We begin with the general definitions for qudits and the stabilizer formalism in higher dimensions in \Cref{sec:qudits}. We then adopt the extended \emph{stabilizer tableau formalism}, which stores the Pauli generator and phase exponents (see \Cref{sub:tableau}). This structure enables direct symplectic manipulation without the explicit matrix exponentiation in Hilbert space.
    \item[Manipulation:] The simulator must implement efficient update rules that describe the behavior of Clifford operations on stabilizers. We derive the full set of gate transformations -- including (qudit generalizations of) $H, S, CNOT, CZ, SW\!AP$ and $X,Z,Y$ -- and prove their symplectic application. The explicit tableau update rules can be found in \Cref{tab:update} and \Cref{sec:simulation}.
    \item[Measurement:] To simulate a computational basis measurement and derive post-measurement states, we generalize the Aaronson-Gottesman CHP algorithm \cite{aaronsonImprovedSimulationStabilizer2004} to prime $d$ dimensions. This method separates deterministic results from random measurement outcomes and performs row operations modulo $d$ to maintain consistency within the tableau. The complete algorithmic description and analysis can be found in \Cref{sub:measurement}. Furthermore, we highlight why composite dimensions are problematic for this kind of measurement algorithm in \Cref{sub:composite} and offer a possible solution.
    \item[Noise:] Realistic simulations require the inclusion of noise. We model such effects with Weyl operator channels, which describe noise as probabilistic mixtures of generalized Pauli errors. The theoretical foundations and sampling strategy are presented in \Cref{sec:noise}.
    \item[Application:] Lastly, we discuss possible applications in \Cref{sec:applications}. 
    Here we describe how the developed noise pushing algorithm can be used for the computation of circuit fidelity \cite{emerson2005scalable}. We also list promising future applications in the conclusion and outlook section.
\end{description}
The remainder of this work is dedicated to the refinement of these ideas.

\paragraph{Outline}
The remainder of this manuscript is organized as follows.
In \Cref{sec:qudits} we review the stabilizer formalism extended to multi-qudit systems. \Cref{sec:simulation} uses this formalism to efficiently simulate Clifford circuits (Gottesman-Knill) as well as (partial and complete) single- and multi-shot measurements. \Cref{sec:code} presents our code -- the \emph{QuickQudits} python library -- and \Cref{sec:noise} discusses different ways to simulate noise in Clifford circuits. 
In \Cref{sec:applications} we provide efficient circuit fidelity estimation as a concrete application and then conclude the manuscript with a discussion in \Cref{sec:discussion}.
    
\section{Qudit Stabilizer Formalism}
\label{sec:qudits}
Information carriers in quantum computers are physical systems that, in general, exist in a high-dimensional Hilbert space of dimension $d$. Consequently, logical information can be stored using $d$ levels in so-called \emph{qudits}, whereas the more widely researched \emph{qubits} are just a special instance with $d=2$.

A single qudit inhabits a Hilbert space $\mathcal{H}_d$ spanned by the orthonormal computational basis 
$
\{ \ket{0}, \ket{1}, \dots, \ket{d-1}\}
$
such that any state $\ket{\psi} \in \mathcal{H}_d$ can be written as a superposition
$
\ket{\psi} = \sum_{i=0}^{d-1} \alpha_i \ket{i}$ with complex-valued amplitudes $\alpha_i \in \mathbb{C}^d$ that obey $ \sum_i |\alpha_i|^2 = 1.
$

Quantum registers over multiple qudits can be built using the Kronecker product (or tensor product).
Hence a composite system over $n$ qudits lives in the Hilbert space $\mathcal{H}_d^{\otimes n}$ and has dimension $d^n$.\\
\subsection{Generalized Pauli and Weyl Operators}
Unitary evolution operators on qudits are $d \times d$ matrices $U$, obeying $U^\dagger = U^{-1}$. A special set of quantum operators are the (generalized) \emph{Pauli matrices} \cite{pauli1927quantenmechanik}. In dimension $d$, they are defined as the \emph{shift} matrix $X$ and the phase (or \emph{boost}) matrix $Z$ 
\begin{align*}
&X\ket{j} = \ket{j+1 \bmod d}\\
&Z\ket{j} = \omega^{j}\ket{j}, \quad \omega = e^{\frac{2\pi \mathrm{i}}{d}}.
\end{align*}
We introduce a secondary phase $\tau = \exp\!\big({\frac{\pi \mathrm{i} (d^2+1)}{d}}\big)$ with $\tau^2  = \omega$ \cite{kueng2015qubit, gross2006hudson}. 
The Pauli-$Y$ operator is defined as $Y = \tau^{-1}XZ$.

We can extend these elementary Pauli matrices by using the general Weyl operator notation
$$
W(a,b) = \tau^{-ab}X^aZ^b, \quad a,b \in \mathbb{Z}_d.
$$
For dimension $d$, there are $d^2$ such Weyl operators. We provide an example for $d=3$ in \Cref{tab:weyl} below.
The qudit Pauli group on $n$ qudits can be then defined as
$$
\mathcal{P}_d^{(n)} = \left\{ \tau^{k} W^{(n)}(\mathbf a,\mathbf b) \; | \;\mathbf  a, \mathbf b \in \mathbb{Z}_d^{(n)}, k \in \mathbb{Z}_{2d} \right\}
$$
where $W^{(n)}(\mathbf a,\mathbf b) = \bigotimes_{i=0}^{n-1} W(a_i, b_i)$.

\begin{table*}[t!]
\centering
\caption{Weyl Operators $W(a,b) = \omega^{ab}X^a Z^b$ for dimension $d = 3$, here $\tau^{-ab} = \omega^{ab}$.}
\label{tab:weyl}
\[
\setlength{\arraycolsep}{5pt}
\renewcommand{\arraystretch}{1.0} 
\begin{array}{ccc}
W(0,0) = 
\begin{pmatrix} 
1 & 0 & 0 \\ 
0 & 1 & 0 \\ 
0 & 0 & 1 
\end{pmatrix} 
& 
W(0,1) = 
\begin{pmatrix} 
1 & 0 & 0 \\ 
0 & \omega & 0 \\ 
0 & 0 & \omega^2 
\end{pmatrix} & 
W(0,2) = 
\begin{pmatrix} 
1 & 0 & 0 \\ 
0 & \omega^2 & 0 \\ 
0 & 0 & \omega 
\end{pmatrix} 
\\[1.6em] 
W(1,0) = 
\begin{pmatrix} 
0 & 0 & 1 \\ 
1 & 0 & 0 \\ 
0 & 1 & 0 
\end{pmatrix} & 
W(1,1) = 
\begin{pmatrix} 
0 & 0 & 1 \\ 
\omega & 0 & 0 \\ 
0 & \omega^2 & 0 
\end{pmatrix} & 
W(1,2) = 
\begin{pmatrix} 
0 & 0 & 1 \\ 
\omega^2 & 0 & 0 \\ 
0 & \omega & 0 
\end{pmatrix} 
\\[1.6em] 
W(2,0) = 
\begin{pmatrix} 
0 & 1 & 0 \\ 
0 & 0 & 1 \\ 
1 & 0 & 0 
\end{pmatrix} & 
W(2,1) = 
\begin{pmatrix} 
0 & 1 & 0 \\ 
0 & 0 & \omega \\ 
\omega^2 & 0 & 0 
\end{pmatrix} & 
W(2,2) = 
\begin{pmatrix} 
0 & 1 & 0 \\ 
0 & 0 & \omega^2 \\ 
\omega & 0 & 0 
\end{pmatrix}
\end{array}
\]
\end{table*}

\paragraph{Commutation Relations:}
We generalize the elementary commutation relation $ZX = \omega XZ$ to 
$$
W(a,b)W(a',b')=\tau^{a'b-ab'}W(a+a',b+b').
$$
This formula shows that Weyl operators form a projective representation of the additive Abelian group $\mathbb{Z}_d^2$. Recall that a set $\mathcal{A}$ together with a binary operation $*$ is a group if the operation is associative and the set has an identity element and an inverse element. The group is \emph{Abelian} if the operation is also commutative, i.e.\ $w_1 * w_2 = w_2 * w_1$ for all $w_1,w_2 \in \mathcal{A}$. The Weyl operators themselves do not commute, but this non-commutativity is fully captured by the projective phase factor.

The product of two Weyl operators follows the group addition law up to a phase factor, which is determined by the antisymmetric exponent $a'b - ab'$.

\paragraph{Symplectic Form:}

To capture the algebraic relationship that defines the phase factor, we introduce the symplectic form on the vector space $\mathbf{v} = (a, b) \in \mathbb{Z}_d^2$:
$$\langle\langle (a,b), (a',b') \rangle\rangle = a'b - ab' \pmod d.$$
Hence, swapping the order of the operators introduces a phase factor determined by the symplectic form,
$$W(a,b)W(a',b') = \tau^{\langle\langle (a,b), (a',b') \rangle\rangle} W(a',b')W(a,b).$$
In words: Two operators $W(\mathbf{v})$ and $W(\mathbf{v}')$ commute if and only if their symplectic form is zero, $\langle\langle \mathbf{v}, \mathbf{v}' \rangle\rangle = 0 \pmod d$. This symplectic structure is fundamental for defining stabilizer groups in the next section, as all elements in this group must commute.

\subsection{Stabilizer Group}
The stabilizer formalism~\cite{gottesman1997stabilizer} provides a compact and highly structured framework for describing a broad class of quantum states, namely those that can be defined as common $+1$ eigenstates of a commuting set of Pauli operators.

\begin{definition}[Stabilizer]
A unitary matrix $U$ is a \emph{stabilizer} of quantum state $\ket{\psi}$ if and only if 
$$U\ket{\psi} = \ket{\psi}.$$ This is the case if $\ket{\psi}$ is a $+1$ eigenstate of $U$. 
\end{definition}
Stabilizers are typically elements of the Weyl group and unitary, so their eigenvalues are well-defined phase factors lying on the complex unit circle. This ensures meaningful projective measurements. Their products are also unitary, preserving the group structure even under Clifford evolution.
\begin{definition}[Stabilizer Group]
The \emph{stabilizer group} $\mathcal{S}$ of a state $\ket{\psi}$ is formed by all Weyl matrices that stabilize this state. This is an Abelian subgroup $\mathcal{S}^{(n)}$ of the full Weyl group $\mathcal{W}_d^{(n)}$.
$$\mathcal{S}^{(n)} \!\subset \! \mathcal{W}_d^{(n)} = \left\{ \bigotimes_{i=0}^{n-1} \tau^{k_i}\!X^{a_i}\!Z^{b_i} \; 
\Big| \; a_i,\!b_i \! \in \! \mathbb{Z}_d, 
\; k_i \!\in \!
\mathbb{Z}_{2d} \right\}\!.$$
In odd dimensions, $k$ is always even and we can replace $\tau^k$ with $\omega^{k/2}$.
\end{definition}
Such a set uniquely identifies a state and it can be efficiently characterized by a set of $n$ \emph{stabilizer generators} $\mathcal{G}= \left\{g_0,\ldots,g_{n-1} \right\}$ denoted as $\langle \mathcal{G} \rangle $. These then generate the full stabilizer group comprised of $d^n$ elements via multiplication.

\begin{example}
The qubit ($d=2$) Bell state $\ket{\Phi_2^+} = \frac{1}{\sqrt{2}} (\ket{00} + \ket{11})$ is stabilized by $\{ II, XX, ZZ, -YY \}$ and the generators are $\langle XX, ZZ \rangle $.
\end{example}
\begin{table*}[t!]
\centering
\caption{The full stabilizer tableau includes generators $\langle g_0 \dots g_{n-1} \rangle$ for the \emph{destabilizers} and generators $\langle g_n \dots g_{2n-1} \rangle$ for the \emph{stabilizers}. The left (red) block holds the exponents for the Pauli-$X$ matrices, the middle (blue) block holds the exponents for the Pauli-$Z$ matrices and the right (yellow) column stores the exponents for the $\tau$ phase. In odd $d$ each $r_i$ must be even.}
\label{tab:tableau}
\[
\left.
\begin{array}{l}
\text{destabilizers}  \\[50pt]
\text{stabilizers}
\end{array}
\right.
\begin{array}{c}
g_0  \\[4pt]
\vdots \\[4pt]
g_{n-1} \\ [4pt]
g_n  \\[4pt]
\vdots \\[4pt]
g_{2n-1} \\
\end{array}
\!\!\!
\left[
\begin{array}{
*{3}{>{\columncolor{red!10}}c}
!{\vrule width 1pt}
*{3}{>{\columncolor{blue!10}}c}
!{\vrule width 1pt}
>{\columncolor{yellow!10}}c
}
x_{0,0} & \cdots & x_{0,n-1} & z_{0,0} & \cdots & z_{0,n-1} & r_0 \\[4pt]
\vdots & \ddots & \vdots & \vdots & \ddots & \vdots & \vdots \\[4pt]
x_{n-1,0} & \cdots & x_{n-1,n-1} & z_{n-1,0} & \cdots & z_{n-1,n-1} & r_{n-1} \\
\noalign{\vskip 2pt}\hline\noalign{\vskip 2pt}
x_{n,0} & \cdots & x_{n,n-1} & z_{n,0} & \cdots & z_{n,n-1} & r_n \\[4pt]
\vdots & \ddots & \vdots & \vdots & \ddots & \vdots & \vdots \\[4pt]
x_{2n-1,0} & \cdots & x_{2n-1,n-1} & z_{2n-1,0} & \cdots & z_{2n-1,n-1} & r_{2n-1}
\end{array}
\right]
\]
\end{table*}
\subsubsection{Destabilizers}
Destabilizers complement stabilizers by adding structure and information. While stabilizers capture the effect of Clifford operations on a state, the combined set of stabilizers and destabilizers forms a Weyl-basis, so a Clifford gate is uniquely defined by its action on both. Tracking these updates is discussed in \Cref{sec:simulation}.
\begin{definition}[Set of destabilizers]
For a given stabilizer group $\mathcal{S}$ with $n$ elements $\{s_0, s_1, \dots, s_{n-1}\} \! \subset \! \mathcal{W}_d^{(n)}$, the associated \emph{destabilizers} are a commuting set of  $n$ Weyl operators $\mathcal{D} = \{d_0, d_1, \dots, d_{n-1}\} \subset \mathcal{W}_d^{(n)}$  such that each element satisfies
$$
s_i d_j = \tau^{k\delta_{ij}} d_j s_i.
$$
That is, each $d_j$ anticommutes with its corresponding stabilizer ($i=j$) and commutes with all others ($i \neq j$).
Here $\delta_{ij}$ is the Kronecker Delta and $k$ is an integer that describes the anticommutation phase. 
\end{definition}

For qubits ($d=2$) the relation above reduces to the standard anticommutation $s_i d_j = -d_j s_i$ \cite{aaronsonImprovedSimulationStabilizer2004}. For qudits, stabilizers and destabilizers also anticommute “up to a phase,” which must be chosen carefully to preserve the tableau structure. Here we set $k=2(d-1)$, giving $\tau^{k}=\tau^{-2}=\omega^{-1}$, so each stabilizer–destabilizer pair satisfies the correct generalized commutation relation. With this choice, $\{\mathcal{S},\mathcal{D}\}$ generates the entire $n$-qudit Weyl group $\mathcal{W}_d^{(n)}$.
The destabilizers are essential for the Clifford circuit simulation described in \Cref{sec:simulation}, as they form the dual basis used in the Aaronson-Gottesman CHP algorithm to compute measurement outcomes and post-measurement updates efficiently. The stabilizer group encodes the quantum state, while the destabilizers carry the missing phase-space information. Together they form a basis of Weyl operators, so that the action of a Clifford gate on both stabilizers and destabilizers uniquely determines that gate. In the stabilizer tableau formalism described in the next section, the combined stabilizer + destabilizer matrix is invertible under symplectic transformations mod $d$; without destabilizers, information is lost and the transformation can no longer be inverted uniquely. Destabilizers also help reconstruct the explicit Pauli expansion of the state.

\subsection{Clifford group}
\begin{definition}[Clifford Group]
The $n$-qudit \emph{Clifford group} $\mathcal{C}_d^{(n)}$ is the set of unitaries that normalize the generalized Pauli or Weyl group
$$
\mathcal{C}_d^{(n)} = \left\{ U \in \mathcal{U}_{d^n} | U \mathcal{P}_d^{(n)} U^\dagger = \mathcal{P}_d^{(n)} \right\}.
$$ 

\end{definition}

\noindent A Clifford gate conjugates any Weyl operator to another Weyl operator,
$$UW(a,b)U^\dagger = \tau^{\Phi(a,b)} W(a',b'),$$
up to a global phase. As a consequence, the action of a Clifford circuit on a stabilizer state can be described by a linear mapping of the stabilizer generators.
If $g \in \mathcal{S}$ stabilizes $\ket{\psi}$, then the transformed state $U\ket{\psi}$ (with $U$ Clifford) is stabilized by $UgU^\dagger$, since
$$g\ket{\psi}=\ket{\psi}\ \Rightarrow\ (UgU^\dagger)\,U\ket{\psi}=U\ket{\psi}.$$
Consequently, the new stabilizer group is
$$
\mathcal{S}' = USU^\dagger = \langle Ug_0U^\dagger, \dots, Ug_{n-1}U^\dagger \rangle
.$$
This insight is commonly referred to as the \emph{Gottesman--Knill theorem} \cite{aaronsonImprovedSimulationStabilizer2004}:
We can efficiently simulate a quantum circuit on classical hardware in polynomial time, if:
\begin{itemize}
	\item[-] The input state is composed of computational basis states.
	\item[-] The circuit only applies Clifford gates.
	\item[-] Measurements at the end are again in the computational basis.
\end{itemize}

\begin{table*}[t!]
\caption{The simulation of Clifford gates on a stabilizer tableau is done by following these update rules. Each operation is applied row-wise to the tableau, the column is given by the qudit index $i$. Any daggered gate applies the same updates with all the additive terms negated. For example $S^\dagger: \mathbf{z_i \leftarrow z_i - x_i}$. The updates on $X$ and $Z$ are computed modulo $d$, phase updates $r$ are modulo $2d$.}
\label{tab:update}
\centering
\setlength\tabcolsep{0.25em}
\begin{tabular}{llll}
\toprule
\textbf{Gate} & \textbf{X update (mod $d$)} & \textbf{Z update (mod $d$)} & \textbf{Phase (r) update (mod $2d$)} \\
\midrule
$I$ & -- & -- & -- \\
$X$ & -- & -- & $\mathbf{r \leftarrow r - 2z_i}$ \\
$Z$ & -- & -- & $\mathbf{r \leftarrow r + 2x_i}$ \\
$Y$ & -- & -- & $\mathbf{r \leftarrow r + 2(x_i - z_i)}$ \\
$W(a,b)$ & -- & -- & $\mathbf{r \leftarrow r + 2(bx_i - az_i)}$ \\
$H$ & $\mathbf{x_i \leftarrow -z_i}$ & $\mathbf{z_i \leftarrow x_i}$ & $\mathbf{r \leftarrow r + 2x_iz_i}$ \\
$S$ & -- & $\mathbf{z_i \leftarrow z_i + x_i}$ & $\mathbf{r \leftarrow r + \Delta}$, where $\mathbf{\Delta} = 
\begin{cases}
\mathbf{x_i}^2, & d \text{ even}\\
\mathbf{x_i}(\mathbf{x_i} - \mathbf{1}), & d \text{ odd}
\end{cases}$ \\
CNOT$_{c,t}$ & $\mathbf{x_t \leftarrow x_t + x_c}$ & $\mathbf{z_c \leftarrow z_c - z_t}$ & -- \\
CZ$_{c,t}$ & -- & $\mathbf{z_c \leftarrow z_c + x_t,\; z_t \leftarrow z_t + x_c}$ & $\mathbf{r \leftarrow r + 2x_cx_t}$ \\
SWAP & $\text{swap}(\mathbf{x_i}, \mathbf{x_j})$ & $\text{swap}(\mathbf{z_i}, \mathbf{z_j})$ & -- \\
\bottomrule
\end{tabular}
\end{table*}

\subsection{Stabilizer Tableaus}
\label{sub:tableau}

The stabilizer tableau representation provides the computational backbone for efficient stabilizer simulation. We encode both stabilizers and destabilizers as rows in a structured matrix, with dedicated columns for $X$-exponents and $Z$-exponents over $\mathbb{Z}_d$ and $\tau$-phase exponents over $\mathbb{Z}_{2d}$. Clifford gates then act by simple algebraic transformations on this matrix which is called a tableau. The tableau not only determines the state through its stabilizer rows but also enables efficient state updates, measurement sampling and noise incorporation without ever constructing explicit exponential-size state vectors or density matrices. The representation was first explained in Ref.~\cite{aaronsonImprovedSimulationStabilizer2004} and generalizes naturally to higher-dimensional qudits \cite{brandl2024efficient, de2011linearized}.
The full visualization for stabilizer tableaus can be found in \Cref{tab:tableau}.
In short, an $n$-qudit stabilizer state can be encoded in a $2n \times (2n + 1)$ integer type tableau
$$T = [X | Z | r]$$
where columns $[0\!:\!n\!-\!1]$ hold $X$ exponents, columns $[n\!:\!2n\!-\!1]$ hold $Z$ exponents and the last column $[2n]$ holds phase exponents $\tau \pmod{2d}$. The top $n$ rows are the destabilizers followed by $n$ stabilizer rows.
The destabilizers are optional for the simple state representation, but become very important in the simulation process and the measurement.

\noindent Let row \(g_i = (\mathbf x_i \mid \mathbf z_i \mid r_i)\), where \(\mathbf x_i, \mathbf z_i \in \mathbb{Z}_d^n\) are row vectors of length \(n\).  
Two rows \(g_i\) and \(g_j\) commute if and only if their symplectic inner product vanishes:
$$
\langle\!\langle g_i, g_j \rangle\!\rangle = \mathbf x_i \cdot \mathbf z_j - \mathbf x_j \cdot \mathbf z_i = 0 \pmod{d}.
$$
This is exactly the symplectic form on the $2n$-dimensional vector space over $\mathbb{Z}_d$, so Clifford updates must preserve this symplectic form to maintain commutation relations between stabilizers.
In higher dimensions, there are more possible ways for a destabilizer $d_i$ and corresponding stabilizer $s_i$ to anticommute. For the simulations, we choose to force them to always anticommute with a phase $\omega^{-1}$:
$$s_id_i = \omega^{-1}d_is_i.$$

\section{Clifford Circuit Simulation}
\label{sec:simulation}

Clifford circuit simulation relies on the fact that Clifford unitaries normalize the Pauli group. Under conjugation, any generalized Pauli operator is mapped to another Pauli operator. This property allows the evolution of a stabilizer state to be tracked entirely through the transformation of its stabilizer generators. In the qudit case, these transformations correspond to symplectic linear mappings over $\mathbb{Z}_d$. By applying the appropriate update rules gate by gate, we obtain a complete and efficient classical simulation algorithm for Clifford circuits, including entangling operations, multi-qudit interactions and measurements. This section provides the explicit tableau update rules and contains a discussion of potential limitations depending on the qudit dimension, followed by an algorithmic complexity analysis.

\subsection{Update Rules}
\label{sub:updates}

Tableau updates can be performed purely algebraically over $\mathbb{Z}_d$, avoiding explicit $d^n \times d^n$ matrices and can be found in \Cref{tab:update}, see also \cite{gheorghiu2014standard} . Each gate operation acts row-wise on the tableau, and $X$ and $Z$ entries are computed modulo $d$; instead of the more common $\omega$ phase, we use $\tau$ phases modulo $2d$ to better include even dimensions.

\subsection{Composite dimensions}
\label{sub:composite}
For composite dimensions $d$, the arithmetic structure of $\mathbb{Z}_d$ changes fundamentally:
it is no longer a field, but a ring and therefore many elements become zero divisors. 
As a consequence, several properties that hold in prime dimensions fail in the composite case. In particular, stabilizer groups may exhibit degeneracies, such as repeating eigenvalues or generators with non-trivial kernel, such that more than $n$ generators may be required to uniquely describe a state \cite{gheorghiu2014standard}.

\noindent For the simulation of Clifford circuits, composite dimensions do not pose any difficulties, since the update rules require only modular addition and do not rely on invertibility.
The measurement procedure in the next section relies on multiplicative inverses and Gaussian elimination modulo $d$, which can fail to produce valid pivots in composite dimensions. The linearized stabilizer formalism (Sec.~IV of \cite{de2011linearized}) addresses this by using an improper extended tableau, allowing updates in composite dimensions.

However, we can still implement a sampling algorithm for composite dimensions and it can be found in \Cref{sub:measurement}. We achieve this by formulating the sampling procedure in terms of exact Smith normal form (SNF) decompositions \cite{hostens2005stabilizer}. Other simulators like Sdim \cite{kabir2025sdim} use ring-based elimination or exact Diophantine solving. Specifically, we rely on SymPy's \cite{10.7717/peerj-cs.103} \texttt{smith\_normal\_decomp} function, which computes the exact SNF together with the corresponding unimodular transformations -- this guarantees correctness for arbitrary composite $d$.
Consequently, the simulator is restricted to exact sampling of measurement outcomes rather than constructing post-measurement tableaus. Noise modeling, however, does not rely on tableau reduction and remains fully supported in all dimensions, as discussed in \Cref{sec:noise}.
\subsection{Measurement/Sampling}
\label{sub:measurement}

Measurement plays a central role in stabilizer-based simulation because it updates not only the state but also the algebraic constraints that define it. The outcome of a measurement depends on the commutation relationship between the measured Pauli operator and the current stabilizer generator. If it commutes with all stabilizers, the outcome is deterministic; if it anticommutes with at least one, the result is uniformly random over $\mathbb{Z}_d$ in prime $d$.
The CHP algorithm, originally described by Gottesman and Aaronson in \cite{aaronsonImprovedSimulationStabilizer2004}, provides an efficient method for simulating Clifford circuits on qubits, including projective measurements in the computational basis. 
\end{multicols}
\begin{algorithm}
\caption{Generalized Prime-Dimension Qudit Measurement Algorithm}
\label{algo:chp}
\begin{enumerate}
	\item Perform a projective measurement of a Pauli operator $M$ on qudit $i$. For the remainder of this paper, we restrict the measurements to the computational ($Z$) basis.
	\item \textbf{Case 1 (Random outcome):}  If $M$ anticommutes with any stabilizer $s_j$, the measurement outcome is random over the $d$ possible eigenvalues of $M$.
    This occurs, if there exists a stabilizer row with $x_{ji} \neq 0$.
	\begin{enumerate}
		\item Choose the pivot row $p$ as the first stabilizer row $s_p$ where $x_{pi} \neq 0$ (anticommutes with $M$).
        \item Compute the modular inverse $m = x_{pi}^{-1} \pmod{d}$ and normalize the pivot row $s_p \leftarrow s_p^{m}$ such that $x_{pi} = 1$.
		\item Eliminate every other anticommuting (destabilizer and stabilizer) generator row $g_q$ via row operations $$g_q \leftarrow g_q \cdot s_p^{-x_{qi}}.$$
		\item Choose a random integer $k \in \mathbb{Z}_d$ as the measurement outcome.
		\item Update the post-measurement tableau: 
		\begin{itemize}
            \item[-] Replace the original destabilizer $d_p$ associated with $s_p$ by setting $d_p = s_p$.
			\item[-] Introduce a new stabilizer row representing the measurement result: $$s_{\text{new}} = \tau^{-2k} M$$ where $M = Z_i$ is the measurement operator and $k$ the measurement outcome.
            \item[-] Replace the measured row with $s_p \leftarrow s_{\text{new}}$.
		\end{itemize}
	\end{enumerate}
	\item \textbf{Case 2 (Deterministic outcome):} 
    If $M$ commutes with all stabilizer $s_j$, the measurement outcome is deterministic and we only need to determine it.
    \begin{enumerate}
        \item Multiply all stabilizer rows that have an associated destabilizer row that anticommutes with $M$.
        $$g = \prod_{\substack{j=0 \\ x^{(d_j)}_{ji}\neq 0}}^{n-1} s_j^{x_{ji}^{(d_j)}}$$
        The resulting row has a single $z_{ji} = 1$ and the deterministic measurement result is encoded in its phase. 
        \item Compute the measurement outcome as $k = -\frac{\tau(g)}{2} \pmod{d}$.
        \item There is no need to update the tableau, as the deterministic result has no influence on other measurements.
    \end{enumerate}
    \item Return the measurement outcome and repeat for each measured qudit $i$.
\end{enumerate}
\end{algorithm}
\begin{multicols}{2}
We generalized this algorithm to qudits of prime dimension $d$ by using our extended tableau formalism (\Cref{sub:tableau}) and the generalized Clifford gate operations (\Cref{sub:updates}).
The stabilizer group defines the subspace of the current quantum state and provides the constraints that the state must satisfy. However, when we apply a Clifford gate or perform a measurement, these constraints change and we must update the stabilizers (and destabilizers) to reflect the new state.
The destabilizers serve as a complementary set of generators and track how stabilizers transform under Clifford operations and are crucial for determining the deterministic measurement outcome. \\
In \Cref{algo:chp} we describe the general structure of our extended measurement algorithm. \\
For composite dimensions, the Smith normal form (SNF) approach allows exact sampling of measurement outcomes but does not produce an updated post-measurement tableau. The procedure, shown in \Cref{algo:snf} and following \cite{hostens2005stabilizer}, enables the sampling of computational-basis outcomes for arbitrary qudit dimensions while respecting all stabilizer constraints. These results can then be used for Monte Carlo simulations, noise analysis, or benchmarking, even when tableau reduction is not feasible.
\end{multicols}
\begin{algorithm}
\caption{Composite-Dimension Qudit Measurement via Smith Normal Form}
\label{algo:snf}
\begin{enumerate}
    \item Identify the \emph{diagonal subgroup} of the stabilizer group:
    \begin{enumerate}
        \item For integer coefficient vectors $\mathbf u \in \mathbb{Z}_d^n$, a product of stabilizer generators is \emph{diagonal} iff its $X$-part cancels:
        $$
        X_s^T \mathbf u \equiv 0 \pmod{d},
        $$
        where $X_s$ are the rows of the stabilizer $X$-block.
        \item Compute a set of kernel generators $\{\mathbf u^{(1)},\dots,\mathbf u^{(r)}\}$ for $X_s^T$ over $\mathbb{Z}_d$
        (i.e., $X_s^T \mathbf u^{(\ell)} \equiv 0 \pmod d$) via \emph{Smith normal form} (SNF).
        \item Form the corresponding diagonal stabilizer $S(\mathbf u^{(\ell)}) = \prod_{j=0}^{n-1} s_j^{u^{(\ell)}_j}$ for each generator $\mathbf u^{(\ell)}$.
    \end{enumerate}
    \item Convert each diagonal stabilizer into a linear congruence over $\mathbb{Z}_d$:
    $$
    \mathbf b \cdot \mathbf m \equiv -\tau^{(\ell)}/2 \pmod{d},
    $$
    where $\mathbf b \in \mathbb{Z}_d^{n}$ is the $Z$-exponent vector and $\tau^{(\ell)}$ the phase exponent of $S(\mathbf u^{(\ell)})$. The computational-basis measurement outcomes $\mathbf m \in \mathbb{Z}_d^{n}$ must satisfy this relation.
    \item Stack all such constraints to form a system:
    $$
    B \mathbf m \equiv \mathbf c \pmod{d},
    $$
    where $B\in\mathbb{Z}_d^{r\times n}$, and $\mathbf c \in \mathbb{Z}_d^{r}$ contains the phase constraints $-\tau^{(\ell)}/2 \pmod d$.
    \item Compute the SNF of $B$:
    $$
    D = U B V,
    $$
    with unimodular matrices $U, V$ and diagonal $D$. Transform the system to
    $$
    D \mathbf y \equiv U \mathbf c \pmod{d}, \quad \mathbf m = V \mathbf y \pmod{d}.
    $$
    \item Solve the transformed system:
    \begin{enumerate}
        \item For each diagonal entry $s_i \neq 0$, solve $s_i y_i \equiv (U \mathbf c)_i \pmod d$, by reducing modulo $\gcd(s_i, d)$ and computing a particular solution $y_0$.
        \item Identify free variables corresponding to zero entries in $D$, which can take any value in $\mathbb{Z}_d$.
        \item For nontrivial gcds, uniformly sample allowed values consistent with the congruence.
    \end{enumerate}
    \item Map back to the original variables:
    $$
    \mathbf m = V \mathbf y \pmod{d}.
    $$
    \item Repeat for each desired measurement shot to sample outcomes exactly. Return the sampled outcomes for the measured qudits.
\end{enumerate}
\end{algorithm}

\begin{figure}[h]
\centering
\begin{subfigure}[t]{0.3\textwidth}
\centering
\[
\left[
\begin{array}{ccc|ccc|c}
     0 & 0 & 0 & 1 & 0 & 0 & 0 \\
     0 & 0 & 0 & 0 & 1 & 0 & 0 \\
     0 & 0 & 0 & 0 & 0 & 1 & 0 \\ 
\noalign{\vskip 2pt}\hline\noalign{\vskip 2pt}
     1 & 0 & 0 & 0 & 0 & 0 & 0 \\
     0 & 1 & 0 & 0 & 0 & 0 & 0 \\
     0 & 0 & 1 & 0 & 0 & 0 & 0
\end{array}
\right]
\]
\caption{Initial full tableau for $\ket{000}$.}
\end{subfigure}
\hfill
\begin{subfigure}[t]{0.3\textwidth}
\centering
\raisebox{-50pt}{
\begin{quantikz}
\lstick{$\ket{0}$} & \gate{H} & \ctrl{1} & \qw  & \rstick{$q_0$}\\
\lstick{$\ket{0}$} & \qw      & \targ{}  & \ctrl{1} & \rstick{$q_1$} \\
\lstick{$\ket{0}$} & \qw      & \qw      & \targ{} & \rstick{$q_2$}
\end{quantikz}
}
\caption{Circuit for GHZ state.}
\end{subfigure}%
\hfill
\begin{subfigure}[t]{0.3\textwidth}
\centering
\[
\left[
\begin{array}{ccc|ccc|c}
     0 & 0 & 0 & 1 & 0 & 0 & 0 \\
     0 & 1 & 1 & 0 & 0 & 0 & 0 \\
     0 & 0 & 1 & 0 & 0 & 0 & 0 \\ 
\noalign{\vskip 2pt}\hline\noalign{\vskip 2pt}
     2 & 2 & 2 & 0 & 0 & 0 & 0 \\
     0 & 0 & 0 & 2 & 1 & 0 & 0 \\
     0 & 0 & 0 & 0 & 2 & 1 & 0
\end{array}
\right]
\]
\caption{Target tableau for the qutrit GHZ state.}
\end{subfigure}
\caption{Stabilizer tableau evolution for the preparation of the three-qutrit ($d=3$) GHZ state $|\mathrm{GHZ}_3^{(3)}\rangle = \frac{1}{\sqrt{3}}\, (|000\rangle + |111\rangle + |222\rangle )$. Here, $H$ is a generalized Hadamard gate and the entangling gate is an extension of CNOT to qudits.}
\label{fig:qutrit_ghz}
\end{figure}
\begin{multicols}{2}

\begin{example}
\label{ex2}

In this example we simulate the measurement of a GHZ state for $n=3$ qutrits ($d=3$). The GHZ state is prepared in \Cref{fig:qutrit_ghz}.
To simulate a measurement on the given tableau we strictly follow the procedure given in \Cref{algo:chp}. The full code example can be found in the appendix.

Measurements can be done in any order, here we begin with the first qutrit $q_0$.

\definecolor{colpurple}{RGB}{200,180,230}
\definecolor{colpurple2}{RGB}{220,180,200}
\definecolor{lightred}{RGB}{245,210,210}

$$
\setlength{\arraycolsep}{6.5pt}
\begin{array}{cccccccc}
& q_0 & & & q_0&&&\\
&\downarrow & && \downarrow&&&
\end{array}
$$
$$
\renewcommand{\arraystretch}{1.2}
\begin{array}{@{} r @{}}
      d_0 \\
      d_1 \\
      d_2 \\
      \rightarrow s_0 \\
      s_1 \\
      s_2 
\end{array}
\left[
\renewcommand{\arraystretch}{1.2}
\begin{array}{>{\columncolor{colpurple!40}}c c c | >{\columncolor{colpurple!40}}c c c | c}
0 & 0 & 0 & 1 & 0 & 0 & 0 \\
0 & 1 & 1 & 0 & 0 & 0 & 0 \\
0 & 0 & 1 & 0 & 0 & 0 & 0 \\ 
\noalign{\vskip 2pt}\hline\noalign{\vskip 2pt}
\cellcolor{yellow!40} 2 & \cellcolor{lightred!80} 2 & \cellcolor{lightred!80} 2 & 0 & 0 & 0 & 0 \\
\cellcolor{colpurple2!80} 0 & \cellcolor{lightred!80} 0 & \cellcolor{lightred!80} 0 & 2 & 1 & 0 & 0 \\
\undermat{X\text{-Block}}{\cellcolor{colpurple2!80} 0 & \cellcolor{lightred!80} 0 & \cellcolor{lightred!80} 0} & 0 & 2 & 1 & 0
\end{array}
\right]
$$
\\
\\
To determine the measurement case, we scan the $X$-block of the stabilizer part and find a non-zero value for $q_0$ at $x_{3,0} = 1$ (highlighted in yellow). This indicates that we are in \textbf{Case 1}, corresponding to a random measurement outcome.
We normalize the selected row such that the pivot element $x_{3,0}$ becomes $1$ and then use this pivot row to eliminate other non-zero entries in the $q_0$ column of the stabilizer $X$-block. 

Next, we create a new stabilizer row for the measurement operator of the form $[000|100|r]$ where the phase entry is determined by the measurement outcome. Assume we measure $m_0 = 2$, then we set $r=2$. The stabilizer row $s_0$ is replaced by this measurement row $[000|100|\textcolor{red}{2}]$ and the associated destabilizer row $d_0$ is replaced by the pivot row $[111|000|0]$.
This yields the post-measurement tableau

$$
\renewcommand{\arraystretch}{1.2}
\begin{array}{@{} r @{}}
      \rightarrow d_0 \\
      d_1 \\
      d_2 \\
      \rightarrow s_0 \\
      s_1 \\
      s_2 
\end{array}
\left[
\renewcommand{\arraystretch}{1.2}
\begin{array}{c c c | c c c | c}
 \rowcolor{yellow!40} 1 & 1 & 1 & 0 & 0 & 0 & 0 \\
0 & 1 & 1 & 0 & 0 & 0 & 0 \\
0 & 0 & 1 & 0 & 0 & 0 & 0 \\ 
\noalign{\vskip 2pt}\hline\noalign{\vskip 2pt}
 \rowcolor{yellow!40} 0 & 0 & 0 & 1 & 0 & 0 & 2 \\
 0 & 0 & 0 & 2 & 1 & 0 & 0 \\
 0 &  0 & 0 & 0 & 2 & 1 & 0
\end{array}
\right]
.$$
\newline

\noindent We now proceed to measure the second qutrit $q_1$.

This time we cannot find any non-zero $X$-entry, so we are in \textbf{Case 2} and compute the deterministic measurement outcome.

$$
\quad
\setlength{\arraycolsep}{6.5pt}
\begin{array}{cccccccc}
\hspace{1em} & q_1  & && q_1&&\\
& \downarrow & && \downarrow&&
\end{array}
$$
$$
\renewcommand{\arraystretch}{1.2}
\begin{array}{@{} r @{}}
      \hspace{1em} d_0 \\
      d_1 \\
      d_2 \\
      s_0 \\
      s_1 \\
      s_2 
\end{array}
\left[
\renewcommand{\arraystretch}{1.2}
\begin{array}{c >{\columncolor{colpurple!40}}c c | c >{\columncolor{colpurple!40}}c c | c}
1 & 1 & 1 & 0 & 0 & 0 & 0 \\
0 & 1 & 1 & 0 & 0 & 0 & 0 \\
0 & 0 & 1 & 0 & 0 & 0 & 0 \\ 
\noalign{\vskip 2pt}\hline\noalign{\vskip 2pt}
 \cellcolor{lightred!80} 0 & \cellcolor{colpurple2!80} 0 & \cellcolor{lightred!80} 0 & 1 & 0 & 0 & 2 \\
\cellcolor{lightred!80} 0 & \cellcolor{colpurple2!80} 0 & \cellcolor{lightred!80} 0 & 2 & 1 & 0 & 0 \\
\undermat{X\text{-Block}}{\cellcolor{lightred!80} 0 & \cellcolor{colpurple2!80} 0 & \cellcolor{lightred!80} 0} & 0 & 2 & 1 & 0
\end{array}
\right]
$$
\\
\\
 Thus, we scan the $X$-block of the destabilizer half for non-zero entries.
$$
\renewcommand{\arraystretch}{1.2}
\begin{array}{@{} r @{}}
      \rightarrow d_0 \\
      \rightarrow d_1 \\
      d_2 \\
      s_0 \\
      s_1 \\
      s_2 
\end{array}
\left[
\renewcommand{\arraystretch}{1.2}
\begin{array}{c >{\columncolor{colpurple!40}}c c | c >{\columncolor{colpurple!40}}c c | c}
\cellcolor{lightred!80}  1 &  \cellcolor{yellow!40} 1 &  \cellcolor{lightred!80} 1 & 0 & 0 & 0 & 0 \\
\cellcolor{lightred!80}  0 &  \cellcolor{yellow!40} 1 & \cellcolor{lightred!80} 1 & 0 & 0 & 0 & 0 \\
\cellcolor{lightred!80}  0 & \cellcolor{colpurple2!90} 0 & \cellcolor{lightred!80}1 & 0 & 0 & 0 & 0 \\ 
\noalign{\vskip 2pt}\hline\noalign{\vskip 2pt}
 0 &  0 &  0 & 1 & 0 & 0 & 2 \\
 0 &  0 &  0 & 2 & 1 & 0 & 0 \\
0 &  0 & 0 & 0 & 2 & 1 & 0
\end{array}
\right]
$$
\\
\\
We create a temporary stabilizer row $s_t$ initialized to $[000|000|0]$ and add to it all  stabilizer rows that have a paired destabilizer row with non-zero $X$-value: $s_t \leftarrow s_t +  $\colorbox{orange!30}{$1\cdot s_0$}$ + $ \colorbox{green!30}{$1 \cdot s_1$}. The multiplicative coefficient for each stabilizer row is determined by the $X$-value in the destabilizer.
$$
\renewcommand{\arraystretch}{1.2}
\begin{array}{@{} r @{}}
       d_0 \\
       d_1 \\
      d_2 \\
      \rightarrow s_0 \\
      \rightarrow s_1 \\
      s_2 
\end{array}
\left[
\renewcommand{\arraystretch}{1.2}
\begin{array}{c c c | c c c | c}
  1 &  \cellcolor{orange!30}1 &   1 & 0 & 0 & 0 & 0 \\
  0 &  \cellcolor{green!30}1 &  1 & 0 & 0 & 0 & 0 \\
  0 &  0 & 1 & 0 & 0 & 0 & 0 \\ 
\noalign{\vskip 2pt}\hline\noalign{\vskip 2pt}
 \rowcolor{orange!30} 0 &  0 &  0 & 1 & 0 & 0 & 2 \\
 \rowcolor{green!30} 0 &  0 &  0 & 2 & 1 & 0 & 0 \\
0 &  0 & 0 & 0 & 2 & 1 & 0
\end{array}
\right]
$$
In this case, the final $s_t$ has a phase value of $2$, such that we can derive that the measurement outcome is $m_1=2$.
Because the outcome is deterministic, it does not need to be updated as it does not affect subsequent measurements. 
Finally, we measure the third qutrit $q_2$ which is again deterministic. We find that $x_{0,2}, x_{1,2}, x_{2,2} > 1$, so we combine rows $s_0, s_1, s_2$ to obtain the measurement row $s_t \leftarrow 1\cdot s_0 + 1 \cdot s_1 + 1 \cdot s_2 = [000|000|2]$, again implying a measurement result of $m_2=2$.

The final post-measurement tableau is 
$$
\renewcommand{\arraystretch}{1.2}
\begin{array}{@{} r @{}}
      d_0 \\
      d_1 \\
      d_2 \\
      s_0 \\
      s_1 \\
      s_2 
\end{array}
\left[
\renewcommand{\arraystretch}{1.2}
\begin{array}{c c c | c c c | c}
  1 &  1 &   1 & 0 & 0 & 0 & 0 \\
  0 &  1 &  1 & 0 & 0 & 0 & 0 \\
  0 &  0 & 1 & 0 & 0 & 0 & 0 \\ 
\noalign{\vskip 2pt}\hline\noalign{\vskip 2pt}
 0 &  0 &  0 & 1 & 0 & 0 & 2 \\
 0 &  0 &  0 & 2 & 1 & 0 & 0 \\
0 &  0 & 0 & 0 & 2 & 1 & 0
\end{array}
\right]
$$
and the overall measurement outcome is $\ket{222}$.
\end{example}

\subsubsection{Affine Subspace Sampling}
\label{subsub:affine}
The extended CHP algorithm is the standard procedure for sampling from a stabilizer tableau. However, whenever we need many shots from the same tableau, we can improve the sampling complexity by a factor $n$ by exploiting the structure of stabilizer measurement outcomes.
All possible measurement results over a given stabilizer tableau form an affine subspace over $\mathbb{Z}_d$ \cite{gross2006hudson}:
$$
\mathcal{A} = v_0 + \mathcal{S}, \ \ \mathcal{S} = \mathrm{rowspan}_{\mathbb{Z}_d}(R), \ \ R = \mathrm{rref}_{\mathbb{Z}_d}(X_\mathcal{S}).
$$
Here, $\mathrm{rowspan}_{\mathbb{Z}_d}(R)$ denotes the set of all $\mathbb{Z}_d$-linear combinations of the rows of the matrix $R$, and $\mathrm{rref}_{\mathbb{Z}_d}(X_\mathcal{S})$ is the reduced row echelon form over $\mathbb{Z}_d$, which produces a canonical basis for the row space. This construction allows exact sampling of measurement outcomes by generating vectors in the affine space $\mathcal{A}$.

The affine subspace is fully characterized by a single offset vector $v_0$, which corresponds to one valid measurement outcome, and the linear subspace of the $X$-block in the stabilizer tableau denoted by $\mathcal{S}$. Once the offset vector $v_0$ and a basis for $\mathcal{S}$ are known, we can generate additional measurement outcomes without re-running the whole CHP procedure. To sample a new measurement outcome, we first identify a basis for the row space of the measurement tableau, denoted as $B$. Let $R$ be the reduced row echelon form of the relevant $X$-block, and let $k = \mathrm{rank}(R)$ be the number of independent rows in this basis. We then draw a random coefficient vector $\mathbf{c} \in \mathbb{Z}_d^k$ and compute the sampled outcome as
\[
v = v_0 + B^T \mathbf{c} \pmod{d}.
\]
Here, the dimensionality $k$ corresponds to the number of independent generators of the affine subspace $\mathcal{S} = \mathrm{rowspan}_{\mathbb{Z}_d}(R)$.

This method produces uniformly distributed outcomes over the affine subspace given by the tableau and reduces the per-shot complexity to $\mathcal{O}(k\cdot n)$.

\begin{continueexample}{ex2}
Previously we have computed an initial valid measurement outcome for the qutrit GHZ state as the offset vector $v_0 = \{2,2,2\}$. To determine the full set of  possible outcomes, we define the affine measurement subspace.
First, we define the basis for the measurement space by calculating the row space of the $X$-block from the original stabilizer tableau. With the reduced row echelon form (rref) we get the independent generators:
$$
\text{rowspan}\!\left(\text{rref}\!\left(\begin{bmatrix}
    2 & 2 & 2\\
    0 & 0 & 0 \\
    0 & 0 & 0
\end{bmatrix}\
\right)\!\right) = \mathcal{S} =  \{1,1,1\}.
$$
Any valid measurement must differ from the offset vector only by a linear combination of the vector $s_0 = \{1,1,1\}$.
Every possible measurement outcome can then be expressed as $v = v_0 + c_is_0$ where $c_i \in\{0,1,2\}$.
By iterating through the coefficients, we get the entire subspace as
\begin{align*}
c &= 0\!:\, v = \{2,2,2\} + 0\cdot \{1,1,1\} = \{2,2,2\} \!\!\!\pmod{3}\\
c &= 1\!:\, v = \{2,2,2\} + 1\cdot \{1,1,1\} = \{0,0,0\} \!\!\!\pmod{3}\\
c &= 2\!:\, v = \{2,2,2\} + 2\cdot \{1,1,1\} = \{1,1,1\} \!\!\!\pmod{3}.
\end{align*}
Hence, all possible measurement results for the qutrit GHZ state are $\{0,0,0\},\{1,1,1\},\{2,2,2\}$.
\end{continueexample}

\subsection{Reduced Tableaus}
A measurement introduces new stabilizer generators and updates all other generators affected by it. After such updates, the tableau may contain linearly dependent stabilizers, trivial generators or rows that no longer contribute to the description of the state. Reduced tableaus address this by systematically eliminating redundant or unnecessary generators to restore a minimal and canonical generator set. This reduction preserves the encoded quantum state while decreasing the tableau size. This is equal to computing the post-measurement state and excluding the qudits that have been measured already.

The general idea is to perform a standard measurement (\Cref{sub:measurement}) and then reduce the resulting tableau by isolating all information about the measured qudit(s) and subsequently remove them from the tableau. This ensures that the post-measurement tableau remains minimal and consistent.
\paragraph{Case 1: Measurement with random outcome.}
If the measurement outcome is random, at least one of the stabilizer generators anticommutes with the measurement operator. Such a generator is then designated as the \emph{measured stabilizer} and replaced by a trivial generator representing the measurement,
$$
g_m = \tau^p \big( I^{\otimes i} \otimes Z \otimes I^{\otimes (n-i-1)} \big).
$$
Here $\tau^p$ with $p \in \mathbb{Z}_{2d}$ encodes the measurement outcome, and the operator only has a single $Z$ on the measured qudit $i$.
In the next step, we want to remove the just measured qudit $i$. For this, we need to incorporate the information contained in the phase of $g_m$ into the remaining stabilizers. 
For any generator $g_j$ where $z_{ji}>0$, we multiply the generators (which is equal to addition in the tableau formalism) $g_j \leftarrow g_j \cdot g_m^c $ where $c=-z_{ji}$ such that $z_{ji} =0$. This step ensures that every stabilizer other than $g_m$ acts trivially on qudit $i$. Moreover, it propagates the correct phase information to the other generators. We don't need to update the destabilizers in this case.
Once all $z_{ji}$ entries have been eliminated, the measurement row $g_m$ and its corresponding destabilizer can be safely removed together with the $\mathbf x_i$ and $\mathbf z_i$ columns. The resulting tableau describes the post-measurement state on $n\!-\!1$ qudits with a size of $2(n\!-\!1) \times (2n\!-\!1)$.

\paragraph{Case 2: Measurement with deterministic outcome.}
For a deterministic measurement, we can directly compute the measurement outcome via a linear combination of the stabilizer generators, such that we obtain a temporary  stabilizer generator 
$$g_m = \tau^p  \big( I^{\otimes i} \otimes Z \otimes I^{\otimes (n-i-1)} \big).$$
This temporary row does not become part of the tableau, but it represents the constraint we need to add to the remaining generators. Similar to the random case, the temporary generator $g_m$ is used to update the other stabilizers such that they act nontrivially on qudit $i$.  With this we ensure that the remaining generators contain the correct phase information. In the deterministic case we also need to explicitly update the destabilizers.
Unlike the random case, no stabilizer is explicitly replaced during measurement. However, some generators become redundant and the remaining rows need to be updated. To detect and remove these, we perform Gaussian elimination modulo $d$ to identify linear dependencies.
Whenever we eliminate a stabilizer row $s_t \leftarrow s_t\cdot s_p^{-\lambda}$ (using pivot row $s_p$), where $\lambda\in\mathbb{Z}_d$ is the elimination coefficient, we apply the dual update to destabilizers $d_p \leftarrow d_p\cdot d_t^{\lambda}$ to preserve commutation relations.
Any stabilizer that becomes a linear combination of the others is discarded, along with its corresponding destabilizer. Finally, the $\mathbf x_i$ and $\mathbf z_i$ columns are removed, producing a reduced tableau of size $2(n\!-\!1) \times (2n\!-\!1)$ for the remaining $n\!-\!1$ qudits representing the post-measurement state.
\newline
\begin{example}
To show the application of the reduced tableau algorithm, we look at a random tableau with $n=3, d=3$

$$
\left[
\begin{array}{ccc|ccc|c}
0 & 0 & 0 & 0 & 2 & 0 & 2 \\
2 & 0 & 2 & 1 & 0 & 2 & 0 \\
0 & 0 & 1 & 0 & 2 & 1 & 4 \\ 
\noalign{\vskip 2pt}\hline\noalign{\vskip 2pt}
1 & 1 & 1 & 0 & 0 & 0 & 0 \\
0 & 0 & 0 & 2 & 1 & 0 & 2 \\
0 & 0 & 0 & 2 & 0 & 1 & 0
\end{array}
\right].
$$
We perform a normal random measurement of qutrit $q_2$. Assume this measures row $s_0$ with the outcome $m_2 = 1$ and the post-measurement tableau

$$
\left[
\begin{array}{ccc|ccc|c}
1 & 1 & 1 & 0 & 0 & 0 & 0 \\
0 & 1 & 0 & 1 & 0 & 2 & 0 \\
2 & 2 & 0 & 0 & 2 & 1 & 4 \\ 
\noalign{\vskip 2pt}\hline\noalign{\vskip 2pt}
 \rowcolor{yellow!40} 0 & 0 & 0 & 0 & 0 & 1 & 4 \\
0 & 0 & 0 & 2 & 1 & 0 & 2 \\
0 & 0 & 0 & 2 & 0 & 1 & 0
\end{array}
\right].
$$

\noindent For the tableau reduction we now need to eliminate all other $Z$-entries in the $q_2$ column. We do this by multiplying $s_2 \leftarrow s_2 \cdot s_0^2= [0,0,0,2,0,1,0] +[0,0,0,0,0,1,4] +[0,0,0,0,0,1,4]= [0,0,0,2,0,0,2]$.

$$
\left[
\begin{array}{ccc|ccc|c}
1 & 1 & 1 & 0 & 0 & 0 & 0 \\
0 & 1 & 0 & 1 & 0 & 2 & 0 \\
2 & 2 & 0 & 0 & 2 & 1 & 4 \\ 
\noalign{\vskip 2pt}\hline\noalign{\vskip 2pt}
0 & 0 & 0 & 0 & 0 & 1 & 4 \\
0 & 0 & 0 & 2 & 1 & 0 & 2 \\
 \rowcolor{yellow!40} 0 & 0 & 0 & 2 & 0 & 0 & 2
\end{array}
\right]
$$
Then we just need to remove the measured stabilizer and destabilizer row together with the columns for $q_2$ and get
$$
\left[
\begin{array}{cc|cc|c}
0 & 1  & 1 & 0  & 0 \\
2 & 2  & 0 & 2  & 4 \\ 
\noalign{\vskip 2pt}\hline\noalign{\vskip 2pt}
0 & 0  & 2 & 1  & 2 \\
0 & 0  & 2 & 0  & 2
\end{array}
\right].
$$
We continue with a deterministic measurement of $q_1$ with value $m_1 = 0$. The CHP algorithm gives us a temporary row here with $t = [0,0,0,1,0]$. Again, we eliminate all other $Z$-entries in the $q_1$ column with $s_1 \leftarrow s_1 \cdot t^2=[0,0,2,1,2] + 2 \cdot [0,0,0,1,0]=[0,0,2,0,2]$ to get
$$
\left[
\begin{array}{cc|cc|c}
0 & 1  & 1 & 0  & 0 \\
2 & 2  & 0 & 2  & 4 \\ 
\noalign{\vskip 2pt}\hline\noalign{\vskip 2pt}
 \rowcolor{yellow!40} 0 & 0  & 2 & 0  & 2 \\
0 & 0  & 2 & 0  & 2
\end{array}
\right].
$$
Then we compute the rref of the stabilizer part to identify the redundant row, which is obvious in this case ($s_2$). We update the paired destabilizer with $d_1 \leftarrow d_1 \cdot d_2 = [0,1,1,0,0] + [2,2,0,2,4] = [2,0,1,2,2]$. Phase is updated as $r\leftarrow r_1+r_2+2(z_1\cdot x_2)\pmod{2d} = 0 + 4 + 4 = 8 \equiv 2 \pmod{6}$, and we obtain
$$
\left[
\begin{array}{cc|cc|c}
 \rowcolor{yellow!40} 2 & 0  & 1 & 2  & 2 \\
2 & 2  & 0 & 2  & 4 \\ 
\noalign{\vskip 2pt}\hline\noalign{\vskip 2pt}
0 & 0  & 2 & 0  & 2 \\
0 & 0  & 2 & 0  & 2
\end{array}
\right].
$$
Now we can remove the redundant stabilizer and destabilizer row ($s_2$ and $d_2$) and also remove the columns for $q_1$ and get the final reduced tableau

$$
\left[
\begin{array}{c|c|c}
2   & 1   & 2 \\
\noalign{\vskip 2pt}\hline\noalign{\vskip 2pt}
0   & 2   & 2 \\
\end{array}
\right].
$$
The stabilizer for the final qutrit is $\tau^2 Z^2$ and is a deterministic $1$.
\end{example}

\paragraph{Application}
Reduced tableaus are particularly useful whenever we want to analyze quantum states under a sequence of measurements, especially in scenarios where qudits are gradually measured and removed from the computation. One key application is \emph{measurement-based quantum computing} (MBQC) \cite{article}, where computation is driven entirely through single-qudit measurements performed on an initially prepared entangled state (e.g.,\ \emph{cluster states}) \cite{PhysRevLett.86.5188}. When a qudit has been measured, the outcome affects the remaining qudits and the computation continues on a strictly smaller subsystem. Full MBQC requires single-qudit rotations as well, which are not Clifford and hence cannot be simulated with stabilizer. However, the reduced tableau offers an efficient way to keep track of an evolving and shrinking subsystem. 

\subsection{Complexity Analysis}
The stabilizer formalism enables efficient classical simulation of Clifford circuits by reducing quantum state evolution to algebraic update rules of a tableau over $\mathbb{Z}_d$. Each Clifford gate corresponds to a linear or affine transformation on the tableau which gives a computational cost proportional to the number of gates $G$ and the number of qudits $n$. Measurement introduces additional overhead: we need to identify the commuting and anticommuting structure among the stabilizers, performing modular Gaussian elimination on a subset of rows and potentially inserting or removing generators. 

In prime dimensions $d$, the measurement of a single qudit runs in $\mathcal{O}(n^2\,\mathrm{polylog}(d))$, where the $\mathrm{polylog}(d)$ factor arises from modular inverse and arithmetic operations on elements of $\mathbb{Z}_d$. Since the qudit dimension $d$ is fixed and typically small in practical applications, we treat these factors as constant and omit them in all subsequent asymptotic complexity expressions for tableau operations. In contrast, Smith normal form computations use exact integer arithmetic, where intermediate values may grow in bit-length. Therefore, we explicitly retain the dependence on $\log d$. Measuring all $n$ qudits with CHP hence has a worst case runtime of $\mathcal{O}(n^3)$, but we only need this once during preprocessing. On average, the procedure can only be marginally accelerated if the tableau is sparse or by employing optimized modular arithmetic for prime dimensions. All further measurement shots require only $\mathcal{O}(k n)$ per shot, giving a total runtime of $\mathcal{O}(S k n)$ for $S$ shots, where $k=\mathrm{rank}(X_S)$ as described in the affine subspace sampling in \Cref{subsub:affine}.

In composite dimensions, measurement sampling is performed using Smith normal form decompositions (\Cref{algo:snf}). Let $T_{\mathrm{SNF}}(n,\log d)$ denote the total runtime spent in SNF decompositions on $n \times n$ integer matrices with entries bounded by $d$ using SymPy's exact arithmetic implementation. Our sampling procedure performs two such SNF decompositions and additional tableau processing requiring $\mathcal{O}(n^3)$ operations to construct the diagonal constraint matrix by combining up to $n$ stabilizer generators across $n$ tableau rows. Thus, the preprocessing cost is $T_{\mathrm{SNF}}(n,\log d) + \mathcal{O}(n^3)$. SymPy uses an exact integer algorithm to compute the SNF. While no explicit asymptotic bound is stated for this implementation, SNF admits polynomial-time algorithms; for example, Kannan-Bachem give an algorithm whose runtime is polynomial in the matrix dimension and the bit-length of its entries (i.e., polynomial in $n$ and $\log d$)~\cite{kannanBachem1979}. After preprocessing, each additional measurement shot requires generating a sample vector and performing a matrix-vector multiplication, giving $\mathcal{O}(n^2)$ per shot. Thus, sampling $S$ shots in composite dimensions requires $T_{\mathrm{SNF}}(n,\log d) + \mathcal{O}(n^3) + \mathcal{O}(S n^2)$.
The full comparison is given in \Cref{tab:runtime}.

\begin{table*}[t!]
\centering
\caption{Comparison of complexities for operations in qudit stabilizer simulation. 
Here, $n$ is the number of qudits, $d$ is the qudit dimension (assumed prime unless noted otherwise), $G$ the number of Clifford gates, and $S$ the number of independent measurement shots in the computational basis.}
\label{tab:runtime}
\vspace{0.5em}
\begin{tabular}{ll}
\toprule
\textbf{Operation} 
& \textbf{Computational Complexity}  \\ 
\midrule

$G$ Clifford updates (all $d$)
& $\mathcal{O}(G\, n)$  \\

Single-Qudit Measurement (Random case) 
& $\mathcal{O}(n^2)$  \\

Single-Qudit Measurement (Deterministic case) 
& $\mathcal{O}(n^2)$  \\

$n$-Qudit Measurement
& $\mathcal{O}(n^3)$ \\

Tableau Reduction (random case) 
& $\mathcal{O}(n^2)$ \\

Tableau Reduction (deterministic case) 
& $\mathcal{O}(n^3)$ \\

Complete Simulation of $G$ Gates and $S$ Measurements 
& $\mathcal{O}(G \, n) + \mathcal{O}(n^3) + \mathcal{O}(S \,k\,n)$ 
\\

Composite-$d$ $n$-Qudit Measurement Sampling
& $T_{\mathrm{SNF}}(n,\log d) + \mathcal{O}(n^3) + \mathcal{O}(S n^2)$ \\
\bottomrule
\end{tabular}
\label{tab:simulation_complexity}
\end{table*}

\section{Code}
\label{sec:code}
\texttt{QuickQudits} is a Python library centered on a stabilizer tableau simulator for qudit Clifford circuits, providing tools for circuit construction and sampling of measurement outcomes.

In addition to the tableau-based simulator, the library includes conventional statevector and density-matrix backends for qudits of arbitrary dimension. They are intended as reference implementation for small instances and are subject to the usual exponential memory limitations.

This section gives a high-level overview of the structure and functionality of the library. Concrete usage examples and visualizations are provided in the appendix. The complete source code, documentation, tests, and illustrative Jupyter notebooks are available at \url{https://github.com/QUICK-JKU/QuickQudits} \cite{zenodo}. All experiments and examples in this work correspond to the release \textbf{v1.0.0}.

\subsection{Circuits and Simulation Backends}

Quantum circuits are represented by the \texttt{QuantumCircuit} class.
A circuit specifies the number of qudits $n$, the qudit dimension $d$, and an ordered sequence of operations, including Clifford gates, measurements and noise channels. Circuits are backend-agnostic: the same circuit object can be executed with different simulation strategies without modification.

The \texttt{QuantumCircuit} class provides various methods for circuit construction and manipulation. Elementary operations include appending gates and noise channels, inserting or deleting operations at arbitrary positions, replacing existing operations, copying circuits, composing and reversing them. Circuits can also be saved to and loaded from disk in the \texttt{.npz} format.

The library also includes a customizable circuit visualization utility. This tool renders circuits with gate symbols, measurement locations, and noise channels. Circuits can be visualized by calling the \texttt{draw} method on a circuit instance. \Cref{fig:placeholder} above depicts two such visualizations.

Circuit execution evolves a backend-specific representation of the quantum state. For stabilizer simulation, the state is represented by a tableau that encodes the generators of the stabilizer group together with the associated destabilizers. In the dense backends, the state is represented explicitly as a statevector or density matrix. The simulator supports both single-shot execution and batched sampling over many shots. In the following sections we focus on the stabilizer tableau backend.

\subsection{Stabilizer Tableau Simulator}

In full mode, stabilizer states on $n$ qudits are represented by a tableau of size $2n \times (2n+1)$ consisting of $2n$ Pauli generators (destabilizers first and then stabilizers). Measurement is possible only in full mode; in non-full mode only the stabilizers are stored. Each row in the tableau stores integer-valued exponents of $X$ and $Z$ together with a phase exponent $\tau$. Phase exponents are stored modulo $2d$, while $X$ and $Z$ exponents are stored modulo $d$.

The simulator provides two tableau implementations. For prime dimension $d$, the \texttt{Tableau} class supports stabilizer evolution, including post-measurement updates and reduced tableaus. For composite $d$, the \texttt{CompositeTableau} class uses the same interface, but does not construct post-measurement or reduced tableaus and only supports direct sampling of measurement outcomes. Both classes share a common API, allowing circuits to be executed independently of whether $d$ is prime or composite.

Circuit execution applies an instance of the \texttt{QuantumCircuit} class to a tableau with the tableau's \texttt{apply\_circuit} method, which sequentially applies Clifford update rules to the tableau. Tableau updates admit \emph{Numba JIT} \cite{10.1145/2833157.2833162} acceleration by default, while a pure \emph{NumPy} \cite{harris2020array} execution path is also available.

After a circuit execution, measurement outcomes can be obtained by sampling from the tableau. The simulator supports single-qudit measurements, joint measurement of all qudits, and partial measurements of arbitrary subsets of qudits, all in single-shot and multi-shot settings. For prime dimensions, measurements update the tableau state and can return reduced tableaus when operating in full mode. For sufficiently small composite Hilbert space dimension $d^n$, we can return the full probability vectors for the measured qudits.
A fully evolved tableau may additionally be converted into a statevector by constructing the common $+1$ eigenstate of the stabilizer generators. This functionality is intended for validation and analysis on small systems.

Finally, the library provides inspection utilities for tableau states based on a textual representation. Tableaus can be rendered in block form, showing the $X$ and $Z$ exponents together with the phase exponent, and stabilizer generators can be printed  in operator form. Intermediate tableau states can also be printed after individual simulation steps during circuit execution. Tableau states can be serialized to and restored from disk using the \texttt{.npz} format.

\subsection{Noisy Circuits}

Noise in \texttt{QuickQudits} is represented using explicit noise channels placed inside a \texttt{QuantumCircuit} instance. Each noise channel references a noise model via an integer model identifier.
Noise models are stored in a global noise registry and define a normalized probability distribution over stochastic error operations. Identical noise models are stored only once and reused across channels.
Noise channels can be appended, inserted, replaced, removed, or cleared. The \texttt{QuantumCircuit} interface supports adding noise at specific circuit locations, inserting noise layers at a given depth, applying noise channels after each gate, shifting all noise to the start or end, and explicitly realizing noise by sampling stochastic errors into concrete Weyl operations.

In the density-matrix backend, noise models additionally define Kraus operators, which are cached after construction. During the simulation, these Kraus operators are applied directly to evolve the quantum state, providing an exact reference implementation of noisy dynamics for small systems.

In the tableau backend, noise can be handled in multiple ways. In the direct mode, noise channels are realized stochastically per shot by sampling Weyl errors and applying them to the tableau. This mode is primarily used for validation. More efficient approaches include noise pushing and Pauli-frame simulation \cite{aaronsonImprovedSimulationStabilizer2004} as described in \Cref{sec:noise}, which provide fast multi-shot sampling with lightweight updates of the tableau phase information.

Noise channels are integrated into the circuit visualization utility. They are rendered as coloured circular shapes with wavy boundaries, where the opacity encodes the associated error probability. An example can be found in the appendix.

Finally, the library provides routines for estimating circuit fidelity with the tableau backend. Both entanglement fidelity and state fidelity can be computed. For verification on small instances, the same fidelities can be cross-checked using density-matrix simulation by executing a composed circuit implementing reference-state preparation, the noisy circuit, the ideal inverse circuit, and the inverse preparation.

\subsection{Utilities, Tests, and Examples}

The library includes a set of modules supporting utility tasks such as random circuit generation, state initialization, probability post-processing, and data handling.

The repository further contains an extensive set of unit tests and a collection of Jupyter notebooks which serve as tutorials for demonstrating typical usage patterns. 

\section{Noisy Simulation}
\label{sec:noise}
Real-world quantum architectures can never be ideal and will inevitably be affected by noise. In the context of Clifford circuits, certain classes of  noise can be effectively modeled by using probabilistic mixtures of Weyl operators. We are able to simulate a variety of noise channels including depolarizing, dephasing, and correlated multi-qudit noise, all of which can be integrated via stochastic simulation \cite{grurl2021stochastic, berquist2022stochastic}. This method maintains the efficiency of stabilizer simulation and allows classical sampling over discrete error distributions. Each error instance then corresponds to the random application of a Weyl operator either prior to, or after a Clifford gate or measurement step. The tableau formalism proves to be a powerful tool in this context, because Weyl errors solely affect the phase column, thus noise can be applied with minimal computational overhead. This section reviews the theoretical foundations of qudit noise channels \cite{dutta2023qudit}, their effects on stabilizers \cite{aigner2025qudit}, and discusses their integration into full circuit simulation.

Let $q_{a,b}$ denote the probability weight of each single qudit Weyl operator with total error probability $p = \sum_{(a,b) \neq (0,0)}q_{a,b}$. Then the completely positive, trace preserving error map acting on a density matrix $\rho$ is
$$
\mathcal{E}(\rho) = (1-p)\rho + \!\sum_{(a,b) \neq (0,0)} \!q_{a,b}W(a,b) \rho W(a,b)^\dagger.
$$
For uniform weights, i.e.\ $q_{a,b}=\frac{p}{d^2-1}$, this reduces to \textit{depolarizing noise}:
\begin{align*}
\mathcal{E}_{\text{dep}}(\rho) 
=& (1-p)\rho + \frac{p}{d^2-1} \!\sum_{(a,b) \neq (0,0)} \!W(a,b) \rho W(a,b)^\dagger \\
=& \left(1 - \frac{d^2}{d^2-1} p\right) \rho + \frac{d}{d^2-1} p \,\frac{\mathbb{I}}{d}.
\end{align*}
This channel drives any state towards the maximally mixed state. In the simulation within the stabilizer formalism, this is done by randomly choosing one of the $d^2 - 1$ non-identity Weyl operators and applying it as a random Weyl error.
Similarly, we can also simulate \textit{dephasing noise}:
$$
\mathcal{E}_{\text{deph}}(\rho) = (1-p)\rho + \frac{p}{d-1} \sum_{b=1}^{d-1} Z^b \rho (Z^b)^\dagger
.$$
This channel preserves diagonal elements of $\rho$ while attenuating off-diagonal entries by a factor $(1-p)$.
Likewise, one can define a generalized dit-flip error for qudits. With probability $p$ it applies a nontrivial shift $X^a, (a=0, \dots, d-1)$ with uniform probability. The corresponding map is
$$\mathcal{E}_{\text{flip}}(\rho) = (1-p)\rho + \frac{p}{d-1} \sum_{a=1}^{d-1} X^a \rho (X^a)^\dagger.$$
This channel permutes the computational basis state in a way that randomizes which state each qudit occupies and preserves the overall population distribution.

Each of the channels above can be incorporated into stochastic Monte Carlo simulations by randomly sampling and applying one of the relevant Weyl operators at each noise location. In the tableau formalism, updating a single Weyl operator is simple because it only affects the phase column of the tableau and leaves stabilizers and destabilizers unchanged. The main source of efficiency comes from \emph{noise pushing}: rather than physically propagating each noise gate through the circuit, the simulator computes the final phase that would result if all noise were pushed to the end. This allows us to capture the combined effect of all Weyl errors in a single tableau phase update, avoiding repeated tableau modifications and enabling fast Monte Carlo sampling.

For composite noise channels or correlated noise models, the Kraus operators can still be decomposed into linear combinations of Weyl operators. We then sample correlated error patterns across multiple different qudits. This setting allows the modeling of several realistic multi-qudit errors such as correlated phase damping or stochastic Weyl errors that change the population of computational basis states, providing a stabilizer-compatible approximation to amplitude damping. We can use Weyl twirling -- a generalization of Pauli twirling --  to transform arbitrary noise into a noise model compatible with stabilizer simulation. In the qubit setting, this is formalized as randomized compiling \cite{PhysRevA.94.052325}, and in upcoming work, we extend this approach to qudits of arbitrary dimension. It maintains the same average fidelity and depolarizing strength as the original error channel but discards non-Pauli components and therefore allows the use of efficient stochastic stabilizer simulation for general physical noise.

\subsection{Noise in Tableaus}
As described before, the simulator supports arbitrary noise models expressed as probabilistic mixtures of Weyl operators. These have simple update rules that only affect the phase column of the stabilizer tableau and leave the $X$- and $Z$-exponent blocks unchanged. This applies to all Pauli and Weyl noise models. We support two noise-simulation methods: Pauli frames and noise pushing.

\subsubsection{Pauli Frames}
Pauli-frame simulation is a standard technique for efficiently handling stochastic Pauli or Weyl noise in Clifford circuits and has been developed in \cite{PhysRevA.99.062337}. Instead of explicitly simulating the noisy circuit, we maintain Pauli frames that store the cumulative effect of the noise. First, the noiseless stabilizer tableau is propagated through the Clifford circuit to obtain the final tableau and a reference measurement outcome. Separately, Pauli frames $W_{\mathrm{pf}} = X^{\mathbf{a}} Z^{\mathbf{b}} = \bigotimes_{j=1}^n X_j^{a_j} Z_j^{b_j}$, where $\mathbf{a},\mathbf{b}\in\mathbb{Z}_d^n$, are initialized with $\mathbf{a}=\mathbf{0}$ and $\mathbf{b}$ sampled uniformly at random in $\mathbb{Z}_d^n$, and then updated by applying the action of each Clifford according to the conjugation rules. Whenever we encounter a noise gate, the sampled Weyl operator is multiplied into the corresponding frame (and not simulated in the tableau). Measurements are performed on the noiseless final tableau only once to obtain a reference outcome, and additional noisy samples are generated by adding the $\mathbf{a}$-vector (the $X$-component of the Pauli frame) modulo $d$ to this reference outcome. This approach avoids modifying the stabilizer generators in the tableau or repeating tableau measurements and allows fast sampling. The full procedure is described in \Cref{algo:pfs}.

\subsubsection{Noise pushing}
Noise pushing exploits the tableau structure. In particular, Weyl noise acts on the stabilizer generators solely through updates of the phase column, leaving the $X$ and $Z$ blocks unchanged.
Instead of tracking Pauli frames through the circuit, this approach can be viewed as propagating the effect of each sampled Weyl noise operator forward through the subsequent Clifford gates and accumulating their net phase contribution at the end of the circuit. Importantly, we do not explicitly construct the corresponding Weyl operators or their distribution, rather, we compute their combined effect on the tableau's phase. Concretely, we simulate the noiseless circuit once and store the intermediate $X,Z$ tableau data immediately before each noise channel. For a given noise realization, this approach produces a single phase contribution vector $\Delta \tau$.
A similar idea exists for graph states, where the noise and the noiseless state are updated independently, as presented in \cite{aigner2025qudit}. That work also shows how to apply the method to stabilizer states (Sec.~VI~A), while the graph-state version is presented for a closed set of update rules.\\
The idea of noise pushing is visualized in \Cref{fig:noise-pushing} and the method for obtaining the final noise contribution vector $\Delta \tau$ is described in \Cref{algo:push}.

Both strategies, Pauli-frame and noise-pushing, avoid repeated simulation of the reference stabilizer tableau. The noiseless circuit is simulated once, and each Monte Carlo shot reduces to propagating a Pauli frame or computing a single phase-update vector. Consequently, sampling from noisy Clifford circuits becomes significantly more efficient, while maintaining the exact distribution from the original noisy circuit.
\end{multicols}

\begin{algorithm}
\caption{Pauli-Frame Simulation for Noisy Clifford Circuits}
\label{algo:pfs}
\begin{enumerate}
\item Run the circuit once while skipping all noise channels to obtain a final noiseless tableau $T_{\mathrm{clean}}$.
\item Measure $T_{\mathrm{clean}}$ in the computational basis once to obtain a reference outcome vector $\mathbf{r}\in\mathbb{Z}_d^n$.
\item Initialize Pauli frames for $S$ shots. For each shot $s\in\{1,\dots,S\}$, initialize exponent vectors $\mathbf{a}^{(s)},\mathbf{b}^{(s)}\in\mathbb{Z}_d^n$ with $\mathbf{a}^{(s)}=\mathbf{0}$ and $\mathbf{b}^{(s)}$ sampled uniformly at random from $\mathbb{Z}_d^n$.
\item Iterate over the circuit operations in order and update the Pauli frames:
\begin{enumerate}
\item If the operation is a noise channel acting on qudit $i$, then for each shot $s$ sample exponents $(\alpha,\beta) \in \mathbb{Z}_d^2$ corresponding to the Weyl operator acting on qudit $i$, $X_i^{\alpha} Z_i^{\beta}$, and update only that qudit in the frame: $\mathbf{a}^{(s)}_i \gets \mathbf{a}^{(s)}_i + \alpha \pmod d$, \quad $\mathbf{b}^{(s)}_i \gets \mathbf{b}^{(s)}_i + \beta \pmod d$.
\item Otherwise (Clifford gate), conjugate each frame $(\mathbf{a}^{(s)},\mathbf{b}^{(s)})$ through the gate using the Clifford conjugation rules.
\end{enumerate}
\item For each shot $s$, output the noisy measurement sample
$\mathbf{m}^{(s)} \gets (\mathbf{r} + \mathbf{a}^{(s)}) \pmod d$.
\end{enumerate}
\end{algorithm}
\begin{figure*}
    \centering
    \begin{subfigure}[t]{0.5\textwidth}
        \centering
        \includegraphics[height=1.0in]{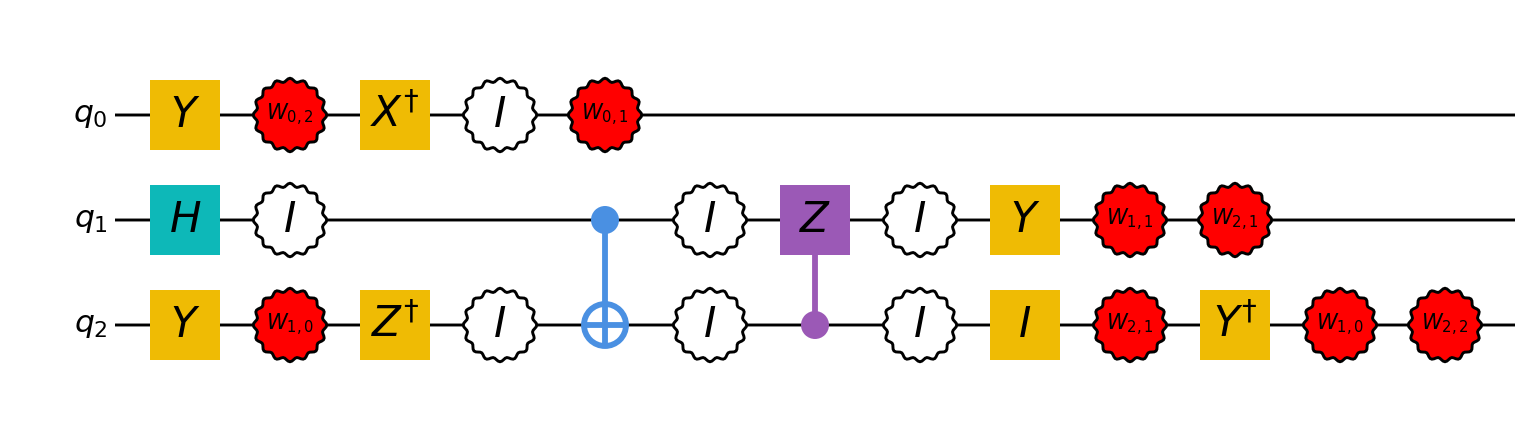}
        \caption{A random noise realization for a given circuit.}
    \end{subfigure}%
    ~ 
    \begin{subfigure}[t]{0.5\textwidth}
        \centering
        \includegraphics[height=1.0in]{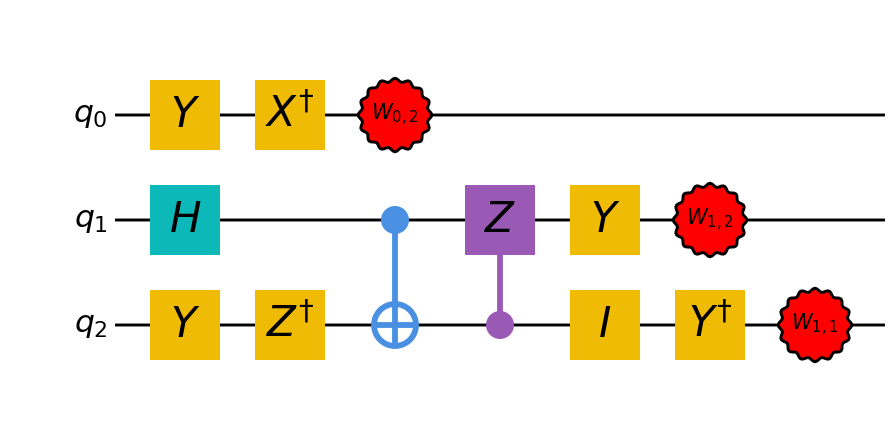}
        \caption{A single Weyl-noise layer at the end of the circuit. The effect is the same as in image (a).}
    \end{subfigure}
    \caption{The noise pushing approach takes a given noise realization and can be viewed as implicitly pushing it to the end of the circuit. We compute a corresponding phase update $\Delta\tau$ from the stored tableau columns and add it to the tableau. Importantly, note that we only compute the phase update and not the explicit final noise layer.}
    \label{fig:noise-pushing}
\end{figure*}
\begin{algorithm}[H]
\caption{Noise pushing in Qudit Clifford Circuits}
\label{algo:push}
\begin{enumerate}
\item Simulate the circuit once while skipping all noise channels (index them by $k=1,\dots,N$). During this simulation,  when the $k$-th noise channel acting on qudit $i_k$ is encountered, store the corresponding tableau columns $\mathbf{x}_{i_k}$ and $\mathbf{z}_{i_k}$ from the current noiseless tableau, and store the noise model (distribution over $(\alpha_k,\beta_k)\in\mathbb{Z}_d^2$) associated with this channel. After the simulation completes, denote the final tableau by $T_{\mathrm{clean}}$.
\item For each shot $s\in\{1,\dots,S\}$:
\begin{enumerate}
\item Initialize $\Delta\tau^{(s)} \gets \mathbf{0}$.
\item For each noise event $k\in\{1,\dots,N\}$:
\begin{enumerate}
\item Sample $(\alpha_k^{(s)},\beta_k^{(s)})\in\mathbb{Z}_d^2$ from the stored noise model of event $k$, corresponding to the local Weyl operator $X_{i_k}^{\alpha_k^{(s)}} Z_{i_k}^{\beta_k^{(s)}}$.
\item Compute the phase contribution using the Weyl update rule
$
\delta\tau_k^{(s)} = 2(\beta_k^{(s)}\,\mathbf{x}_{i_k} - \alpha_k^{(s)}\,\mathbf{z}_{i_k}) \pmod{2d}.
$
\item Update
$
\Delta\tau^{(s)} \gets \Delta\tau^{(s)} + \delta\tau_k^{(s)} \pmod{2d}.
$
\end{enumerate}
\item Form the noisy tableau by updating only the phase column of the clean final tableau:
$T_{\mathrm{noisy}}^{(s)} \gets T_{\mathrm{clean}}$ with phase $\tau \gets \tau + \Delta\tau^{(s)} \pmod{2d}$.
\item Obtain the measurement sample from $T_{\mathrm{noisy}}^{(s)}$.
\end{enumerate}
\end{enumerate}
\end{algorithm}

\begin{multicols}{2}

\subsection{Complexity of Noisy Simulation}
The noise-pushing strategy introduced here yields a significant reduction in simulation cost for noisy Clifford circuits, as it avoids propagating every noise instance. It can also be combined with Pauli-frame simulation for further efficiency in certain workflows.
\paragraph{Preprocessing cost.}
We perform the exact simulation of the noiseless circuit $U$ exactly once. During this simulation, immediately before each noise channel, we store the corresponding $X,Z$ tableau columns that are necessary to calculate the noise contribution later. The simulation follows the Clifford update rules described in \Cref{sec:simulation} and requires $\mathcal{O}(Gn)$ operations for $G$ gates and $n$ qudits, while storing the tableau columns for $N$ noise channels requires $\mathcal{O}(Nn)$ entries. No noise sampling has happened yet. In a Pauli-frame simulation, we likewise simulate the noiseless tableau once, but do not cache intermediate tableau columns, instead propagating the frame through the Clifford circuit separately for each shot.

\paragraph{Cost per noise realization.} 
We use the stored $X,Z$ columns together with a sampled sequence of Weyl operator exponents representing each noise realization to compute the corresponding phase update vector $\Delta\tau$. The update for a single noise operator combines the sampled components $(\alpha_k,\beta_k)$ with the stored tableau columns $(\mathbf{x}_k,\mathbf{z}_k)$ via $2(\beta_k\mathbf{x}_k - \alpha_k\mathbf{z}_k) \bmod 2d$. Summing over all $N$ noise operators gives the total phase update
$\Delta\tau = \sum_{k=1}^{N} 2(\beta_k\mathbf{x}_k - \alpha_k\mathbf{z}_k) \pmod{2d}$, which requires $\mathcal{O}(Nn)$ per shot, since each contribution is a length-$n$ vector. In Pauli-frame simulation, each shot requires propagating the frame through the Clifford circuit and applying the sampled noise operators. Each Clifford gate or noise operator updates only the affected qudit indices of the frame, giving $\mathcal{O}(G+N)$ operations, while initializing the frame requires $\mathcal{O}(n)$ operations. The total per-shot cost is therefore $\mathcal{O}(G+N+n)$.

\paragraph{Measurement and sampling cost.}
After computing the phase update vector $\Delta\tau$, the effect on stabilizer measurements is obtained by updating the phase column of the tableau, which requires $\mathcal{O}(n)$ operations per shot. For prime dimensions $d$, the resulting tableau is measured using the standard CHP measurement algorithm presented in \Cref{sub:measurement}, with runtime $\mathcal{O}(n^3)$ per shot. For composite dimensions, SNF measurement algorithm is used (\Cref{algo:snf}) with runtime $T_{\mathrm{SNF}}(n,\log d) + \mathcal{O}(n^3)$. In Pauli-frame simulation, the noiseless tableau is measured once during preprocessing, and each additional shot only requires adding the frame's $X$-exponent vector to the reference outcome, which adds $\mathcal{O}(n)$ to the runtime.

\paragraph{Total runtime.}
Let $T_{\mathrm{meas}}(n,d)$ denote the runtime required to measure a stabilizer tableau. As stated before, for prime dimensions $d$, $T_{\mathrm{meas}}(n,d)=\mathcal{O}(n^3)$, and for composite dimensions $T_{\mathrm{meas}}(n,d) = T_{\mathrm{SNF}}(n,\log d) + \mathcal{O}(n^3)$. The overall runtime complexity for $S$ Monte-Carlo shots in a circuit with $G$ gates, $N$ noise operators and $n$ qudits is then
$$
\mathcal{O}(Gn) + \mathcal{O}(S(Nn + n + T_{\mathrm{meas}}(n,d)))
$$
for the noise-pushing approach. Pauli-frame simulation requires the noiseless circuit simulation once, a single reference measurement, and then per-shot frame propagation with outcome generation, giving
$$
\mathcal{O}(Gn) + T_{\mathrm{meas}}(n,d) + \mathcal{O}(S(G+N+n)).
$$
Pauli frames, however, only support sampling and do not provide access to the post-measurement tableau or tableau reduction. A comparison is given in \Cref{fig:pf-np}.

\begin{figure*}[t!]
    \centering
    \begin{subfigure}[t]{0.49\textwidth}
        \centering   \includegraphics[width=0.9\textwidth]{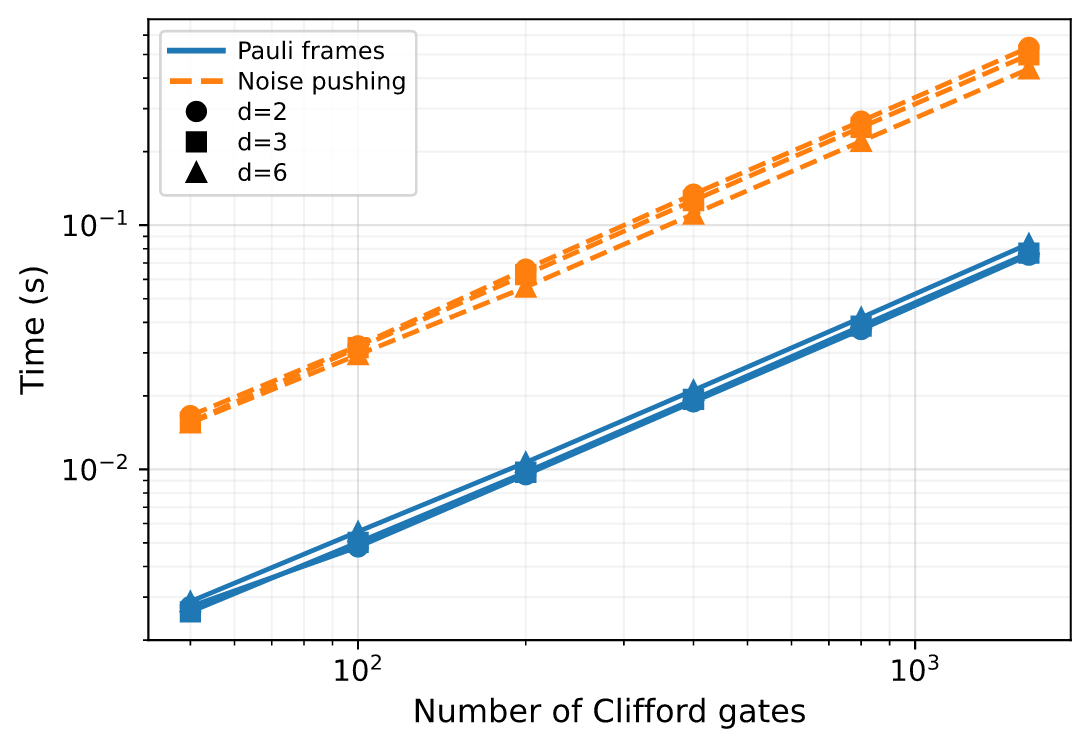}
        \caption{Runtime of Pauli-frame and noise-pushing simulation without measurement as a function of the number of Clifford gates, at fixed system size $n=24$.}
    \end{subfigure}%
    ~ 
    \begin{subfigure}[t]{0.49\textwidth}
        \centering     \includegraphics[width=0.9\textwidth]{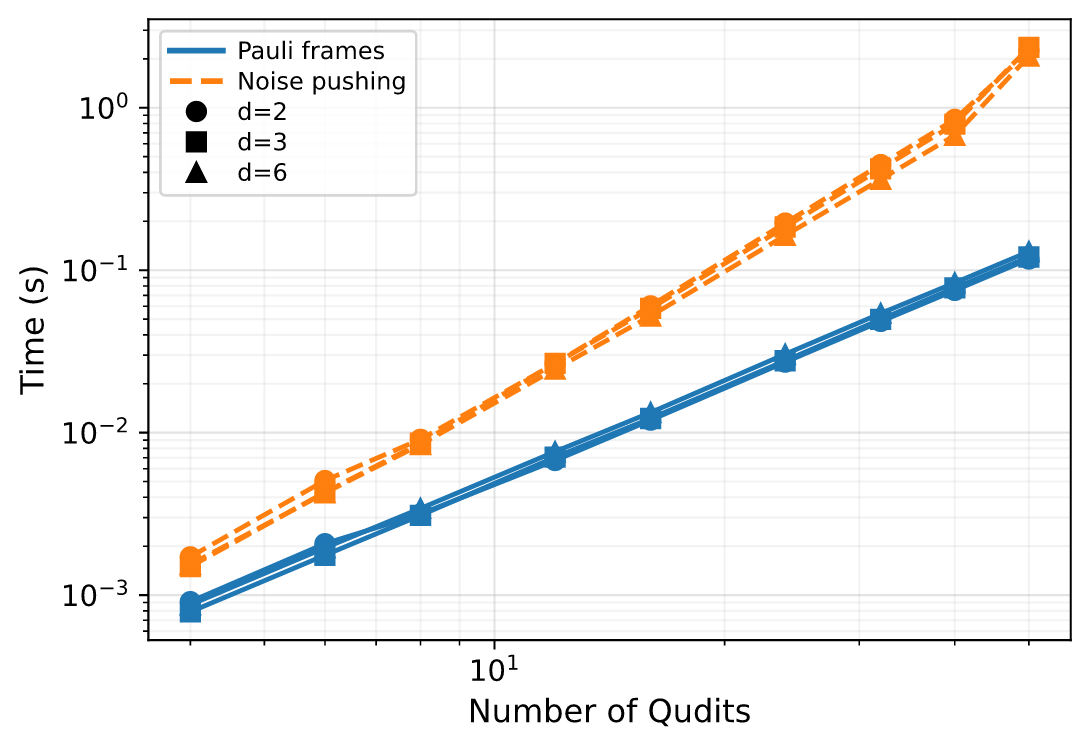}
        \caption{Runtime of Pauli-frame and noise-pushing simulation without measurement as a function of the number of qudits, with circuit size scaling as $G=n^2$.}
    \end{subfigure}
    \caption{Runtimes are averaged over $3$ runs of $10$ randomly generated Clifford circuits with $1000$ Monte Carlo shots each. After every Clifford gate, a single-qudit depolarizing noise channel is applied with error probability $p=0.5$. No measurement is performed, only the final phase-update vector $\Delta\tau$ for noise pushing and the final Pauli-frame exponent vectors $(\mathbf a,\mathbf b)$ for Pauli-frame simulation are computed.}
    \label{fig:pf-np}
\end{figure*}
\section{Application: Efficient Circuit Fidelity Estimation}
\label{sec:applications}
Circuit fidelity quantifies how close the actual noisy output state is to the ideal state produced by a perfect circuit execution. For a fixed circuit that prepares a pure state, this reduces to the state fidelity 
\(
F = \mathrm{tr}(\rho_{\mathrm{ideal}}\, \rho_{\mathrm{noisy}})
\).
A method that is experimentally motivated to estimate such fidelities is the use of \emph{mirror circuits} \cite{proctor2022measuring}. In this protocol, a circuit $U$ is concatenated with a quasi-inverse $\tilde U$ such that, in the absence of noise, the combined sequence returns the system to the input state again (or to a known Pauli-basis state). Additional layers of randomly chosen local Pauli/Weyl or Clifford gates are inserted before, between, and after the forward and reverse parts. The survival probability of the predicted output state then serves as an estimator for the average circuit fidelity.

In our simulation framework, this procedure can be modelled efficiently for Clifford circuits. For a given Clifford sequence, we simulate a noisy forward evolution followed by an ideal reverse evolution and compute the probability of returning to the expected Pauli-basis state. In the noiseless case, the combined circuit returns the system deterministically to that state. By sampling noise realizations in a Monte Carlo scheme and averaging the survival probability, we obtain an estimator for the circuit fidelity. We continue to describe the protocol using the computational zero state $\ket{0}^{\otimes n}$ for clarity, but the same framework can be used for the evaluation of entanglement fidelity by choosing an appropriate entangled input state (e.g., a maximally entangled state with a reference system). 

Moreover, we can systematically vary noise parameters and analyse how the survival probability changes across a parameter range. In the original mirror-circuit protocol, the reported fidelity is typically averaged over randomized Clifford layers and inserted Pauli layers; our simulator allows both fixed-circuit analysis and ensemble averaging over such randomized constructions.

\subsection{Mirror Circuit Simulation}

Given a Clifford circuit $C = C_0 C_1 \dots C_{m-1}$, the mirror-circuit protocol constructs a \emph{mirrored} sequence 
$$
U = C \tilde{C}^{-1},
$$
where $\tilde{C}^{-1}$ denotes a quasi-inverse of $C$. In the absence of noise, the combined circuit maps the computational zero state to a known Pauli-basis state,
$$
U \ket{0}^{\otimes n} = \ket{b}.
$$
This reference state $\ket{b}$ can be determined efficiently by simulating the noiseless circuit once.

On noisy hardware the implemented circuit prepares a state $\rho$, and the protocol estimates the survival probability
$$
p = \langle b | \rho | b \rangle,
$$
i.e., the probability of observing the predicted Pauli-basis outcome. In practice, this probability is estimated by repeatedly executing the same circuit and measuring the fraction of shots returning the expected outcome.

For a single circuit instance, the rescaled survival probability $S$ can be interpreted as the polarization of an effective depolarizing channel. After averaging over the randomized layers used in the mirror-circuit construction, $S$ exhibits an exponential decay with circuit depth. Fitting this decay yields a parameter that provides an estimate of the average entanglement infidelity of the implemented circuit \cite{proctor2022measuring,proctor2022scalable}.

A direct stabilizer-tableau simulation requires two full circuit simulations per noise realization: (1) a forward simulation of the noisy circuit and (2) a backward simulation of the ideal inverse. Even though each run is efficient in theory, this quickly becomes practically challenging for large sequence lengths and many Monte-Carlo samples.

In our simulator, this protocol can be evaluated efficiently for Clifford circuits. We first determine the ideal output state $\ket{b}$ by simulating the noiseless mirrored circuit. Subsequently, Monte Carlo sampling of noise realizations allows us to estimate the survival probability and study how it changes as a function of the noise parameters.

\subsection{Noise Pushing and Compressed Simulation}
The fast tableau simulator developed in this work allows a significant speedup of fidelity simulations because it exploits the fact that Weyl noise acts only on the phase column of the stabilizer tableau. As established in \Cref{sec:noise}, we pretend to combine all possible noise realizations as one single noise frame $\mathcal{N}$ at the end of the ideal circuit $U$. This reproduces the exact effect of distributed noise occurring throughout the noisy circuit $\tilde{U}$
$$
\tilde{U} \rightarrow \mathcal{N}\!\left(U \ket{0}\!\bra{0} U^\dagger \right).
$$

For the following argument, we consider a simplified version of the mirror-circuit protocol in which we only check whether the final state coincides with the input state. This corresponds to a basic subroutine of the full randomized benchmarking procedure.
Rather than simulating the full noisy circuit $\tilde{U}$ explicitly, we simulate the noiseless circuit $U$ once and incorporate the noise via a single stochastic phase update. For each Monte Carlo shot, a noise realization induces an effective phase shift $\Delta \tau$, such that the noisy evolution is captured by this accumulated phase contribution to the the final tableau of $U$.
Equivalently, pushing the noise through the circuit corresponds to considering the conjugated channel
\[
\mathcal{N}'(\cdot) 
= U^\dagger \mathcal{N}\!\left(U (\cdot) U^\dagger \right) U,
\]
which represents the Clifford conjugation of the noise layer.
Consider a noise process $\mathcal{N}(\cdot)$ which is a mixture of Weyl operators
$$
\mathcal{N}(\rho)
=
\sum_{a,b} p_{a,b}\, W(a,b)\,\rho\,W(a,b)^\dagger,
\; W(a,b)=X^a Z^b.
$$
After pushing all noise through the circuit $U$, we obtain a single effective Weyl operator $W_\mathrm{eff} = W(\Delta \tau)$ where the vector $\Delta \tau$ encodes the final phase contribution of all propagated noise. The circuit conjugates 
$$
U^\dagger (X^{\mathbf a} Z^{\mathbf b}) U
    = \tau^\phi X^{\mathbf a'} Z^{\mathbf b'}.
$$
This reduces the problem to checking whether the conjugated noise stabilizes the initial state $\ket{0 \dots 0}$
$$|\psi_{\mathrm{out}}\rangle
    = (U^\dagger W_{\mathrm{eff}} U)\,\ket{0}^{\otimes n}.
$$
This occurs if and only if $(U^\dagger W_{\mathrm{eff}} U)$ is a product of $Z$-type Weyl operators, i.e. if the effective noise acts trivially on the input state which is the case if
$\Delta \tau = 0$. This idea is captured by the following proposition. 

\begin{proposition}
\label{prop:1}
Let $U$ be a (noiseless) Clifford circuit acting on $\ket{0}^{\otimes n}$, then $\mathcal{S}_\mathrm{in} = \langle Z_0, \dots, Z_{n-1} \rangle$ is the stabilizer group of the input and $\mathcal{S}_\mathrm{out} = U\mathcal{S}_\mathrm{in}U^\dagger$ the stabilizer group of output $U\ket{0}$.
Let $W_\mathrm{eff}$ be the single Weyl operator obtained by pushing all noise through $U$ to the last layer and let $\Delta \tau$ denote the phase action of $W_\mathrm{eff}$ on the stabilizer generators $\mathcal{S}_\mathrm{out}$. If $\Delta \tau = 0$ then $W_\mathrm{eff}$ commutes with every generator of $\mathcal{S}_\mathrm{out}$ and therefore 
$$
U^\dagger W_\mathrm{eff}U \in \langle Z_0, \dots, Z_{n-1} \rangle.
$$
In other words, the conjugated noise is a product of $Z$-type Weyl operators and acts trivially on $\ket{0}^{\otimes n}$.

\end{proposition}

\begin{proof}
Assume that $\Delta \tau = 0$, then all output generators commute with the effective Weyl operator
$$
[S, W_\mathrm{eff}] = 0, \quad \forall S \in \mathcal{S}_\mathrm{out}.
$$

Conjugating both operators by $U^\dagger$ for every $S \in \mathcal{S}_\mathrm{out}$ yields
$$
[U^\dagger S U, U^\dagger W_\mathrm{eff} U] = 0
$$
This follows from
\begin{align*}
    (U^\dagger W_\mathrm{eff}U)(U^\dagger SU) &= U^\dagger (W_\mathrm{eff}S) U \\
    &= U^\dagger (SW_\mathrm{eff})U \\
    &= (U^\dagger S U) (U^\dagger W_\mathrm{eff}U).
\end{align*}

Every $U^\dagger S U$ is in $\mathcal{S}_\mathrm{in} = \langle Z_0, \dots, Z_{n-1} \rangle$. The only Weyl operators commuting with all $Z$-type operators are of $Z$-type themselves. Thus $U^\dagger W_\mathrm{eff} U$ is a product of $Z$-operators up to global phase and therefore stabilizes $\ket{0}^{\otimes n}$.

\end{proof}
\paragraph{Remark:}
The converse also holds: if $U^\dagger W_{\mathrm{eff}} U$ is a $Z$-type operator (hence acts trivially on $\ket{0}^{\otimes n}$), then $W_{\mathrm{eff}}$ commutes with all elements of $\mathcal{S}_{\mathrm{out}}$, implying \(\Delta\tau=0\).
The proposition shows that we distinguish two cases:
\begin{itemize}
\item \textbf{Case 1:} $\Delta \tau = 0$.\\
Then the effective Weyl operator commutes with every stabilizer of the output and acts trivially on the input state
$$
U^\dagger W_\mathrm{eff} U \ket{0}^{\otimes n} = \langle Z_0, \dots, Z_{n-1} \rangle \ket{0}^{\otimes n} = \ket{0}^{\otimes n}.
$$
The shot is counted as success.
\item \textbf{Case 2:} $\Delta \tau \neq 0$. \\
Then at least one stabilizer $S \in \mathcal{S}_\mathrm{out}$ anticommutes with $W_\mathrm{eff}$ leading to 
$$
U^\dagger W_\mathrm{eff} U \ket{0}^{\otimes n} =  \ket{\psi} \neq \ket{0}^{\otimes n}.
$$
The shot is counted as failure.
\end{itemize}

\begin{corollary}
Under the assumptions of Proposition~\ref{prop:1}, the success fidelity of a noisy Clifford circuit acting on $\ket{0}^{\otimes n}$ is uniquely determined by the condition $\Delta \tau = 0$.
In particular, it suffices to simulate the noiseless circuit $U$ once to obtain the output stabilizer tableau. For each Monte Carlo shot, we sample a noise realization, compute the corresponding phase update $\Delta \tau$, and check whether $\Delta \tau = 0$. No simulation of the inverse circuit $U^\dagger$ is required.
\end{corollary}

In summary, the compressed-noise tableau simulator provides an efficient and exact method for computing the circuit fidelity on Clifford circuits. It produces the same fidelities as the standard noisy-forward-reverse protocol, but with drastically reduced computational overhead.
We can also determine the circuit fidelity via Pauli frames. There, we can use a similar proof to show that it is enough to check whether the final Pauli frame (which is equivalent to the noise being pushed to the end) is trivial -- i.e. vectors $\mathbf a = \mathbf 0, \mathbf b = \mathbf  0$. In the current implementation the Pauli frame method is faster, because the $\Delta \tau$ computation needs an expensive final recomputation step. However, the noise pushing approach remains more expressive, as it includes all information necessary for strong simulation with post-measurement tableaus. Furthermore, we believe that by employing a symbolic noise pushing strategy, similar as Ref.~\cite{fang2024symphase}, we can possibly improve the runtime. This could result in a more expensive pre-computation step, followed by a faster per-shot runtime and we leave this for future work.

\section{Discussion}
\label{sec:discussion}
Our work presents a unified framework for the efficient classical simulation of Clifford circuits acting on qudits of arbitrary dimension. We incorporate both stabilizer dynamics and probabilistic Weyl-type noise. By extending the tableau formalism to handle noisy-phase tracking and exploiting the algebraic structure of \(\mathbb{Z}_d\), we achieve scalable simulation without constructing full \(d^n \times d^n\) matrices. Still, we preserve exact sampling of measurement outcomes. For prime dimensions, we implement reduced post-measurement tableaus, whereas for composite dimensions we rely on exact Smith normal form decompositions to sample outcomes efficiently. 
\end{multicols}
\begin{figure}[t]
    \centering
    \includegraphics[width=0.75\textwidth]{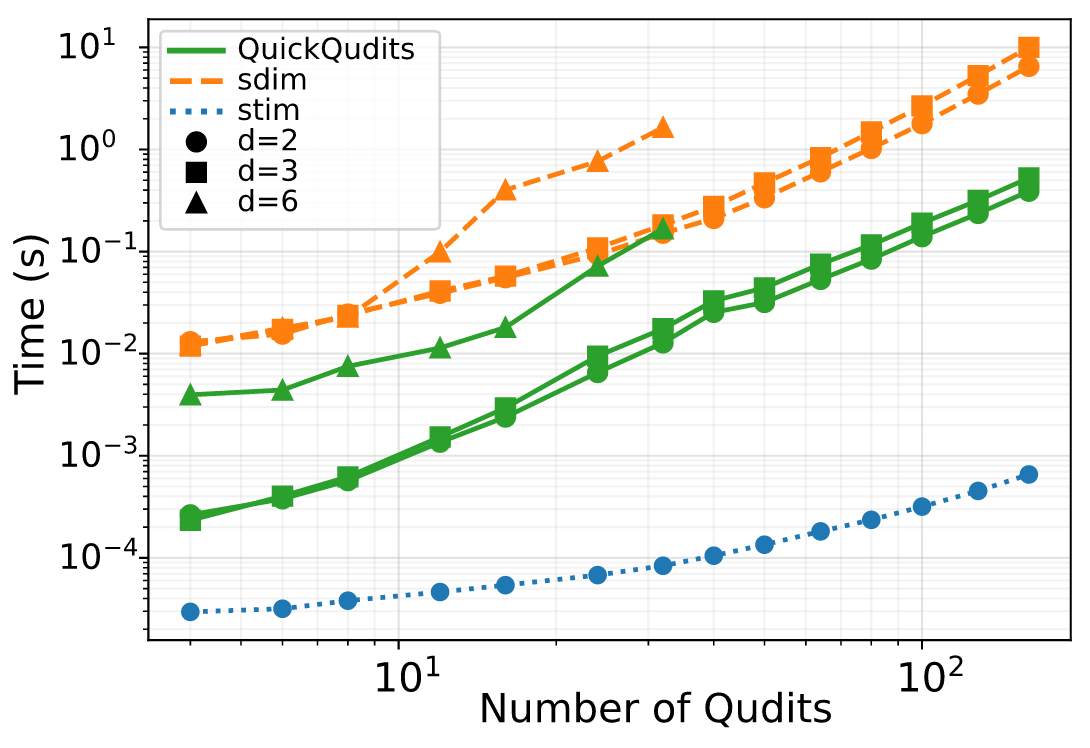}
    \caption{The plot shows the execution time (in seconds) as a function of the number of qudits for QuickQudits (solid green), sdim (dashed orange), and stim (dotted blue). Markers denote the local dimension: circles ($d=2$), squares ($d=3$), and triangles ($d=6$). We introduce a cutoff point at $n=32$ for composite dimensions as the runtime rapidly grows for the Sdim and QuickQudits simulators. While stim provides a baseline for the specialized $d=2$ case, QuickQudits demonstrates a consistent order-of-magnitude speedup over Sdim across all tested dimensions.}
    \label{fig:compare}
\end{figure}
\begin{multicols}{2}
This illustrates how linear-algebraic methods can substitute for standard Gaussian elimination which is used in the generalized CHP algorithm. Noise effects are efficiently captured via Pauli-frame and phase-vector abstractions. In the latter, the cumulative effect of all Weyl errors is recorded as a vector of phase shifts (\(\Delta\tau\)) on the stabilizer and destabilizer rows, without updating the full tableau. Both techniques allow efficient Monte Carlo simulation of circuit fidelity for general Weyl-error channels. This framework generalizes known qubit methods, connects naturally to symplectic structures in the Clifford group, and provides a foundation for future extensions, including randomized compiling for qudits, efficient strong simulation in special cases, and exploration of restricted composite-dimension tableau reductions. Overall, it provides a flexible and computationally efficient toolset for studying stabilizer dynamics, noisy qudit circuits, qudit-based quantum error correction and fault tolerance, as well as related quantum information protocols.

\subsection{Comparison to other approaches}
\label{sub:compare}
A recent qudit stabilizer simulator, Sdim \cite{kabir2025sdim}, is the first broadly available open-source tool targeted at qudit ($d>2$) Clifford circuits and is modeled after Aaronson's CHP and Gidney's \emph{Stim} \cite{gidney2021stim} paradigms, with a focus on general Clifford evolution and sampling correctness for prime and small composite values. However, Sdim's composite-dimension solver is known to have implementation limitations and potential errors in its “fast” composite routines, such as invalid stabilizer tableaus, wrong measurement probabilities or outcomes, and inconsistent stabilizer generators that do not properly represent a valid stabilizer state. In contrast, our framework uses exact Smith normal form (SNF) decompositions for all composite dimensions. This enables exact sampling of measurement outcomes for arbitrary $d$.
We present a runtime comparison in \Cref{fig:compare} where we plot the average runtime over $3$ runs, each $10$ random circuits with $n^2$ gates and $1000$ measurement shots. As expected, Stim performs best for qubits, whereas QuickQudits outperforms Sdim in all dimensions.
More formal frameworks like the \emph{Noisy Qudit Stabilizer Formalism} \cite{aigner2025qudit} provide a theoretical description of stabilizer evolution subject to generalized Pauli-diagonal noise in prime-power dimensions. They use a graph-state representation and can describe noisy state evolution under measurements and Clifford gates, offering an alternative algebraic approach to handling noise in qudit contexts. 
SymPhase \cite{fang2024symphase} introduces symbolic phase expressions into the stabilizer tableau in order to efficiently represent the effect of Pauli noise and conditional corrections without repeatedly resimulating the circuit. Conceptually, SymPhase is designed to support branching and conditional corrections, whereas our framework targets Monte Carlo noise sampling and circuit-level fidelity estimation.

In contrast to these existing tools, our simulator combines broad support for all local dimensions with both Pauli-frame (Weyl-frame) tracking and an accumulated phase-shift vector $\Delta \tau$. Unlike many existing packages, we incorporate reduced post-measurement tableaus in prime dimensions, custom visualization of circuits and tableau structure, and implement a unified noise framework that consistently handles both prime and composite $d$. These features give our simulator greater expressiveness for analyzing stabilizer structure and noise effects than tools that either make compromises in composite dimensions, do not track phases beyond Pauli frames, or are optimized solely for qubits. In addition, the simulator emphasizes usability: it provides a user-friendly interface, step-by-step tutorials, and easy installation, making it accessible for both research and educational purposes.

\subsection{Future Work}
Future work will focus on applying our stabilizer tableau framework to the analysis and benchmarking of structured qudit circuits. With the simulator, we can efficiently track tableau phases under noise and propagate Weyl-type errors, which makes it possible to study how Clifford circuits perform in realistic noisy environments. We aim to extend existing qubit-based protocols to higher dimensions, including fixed-size GHZ state preparation \cite{quek2024multivariate}, (multi-qudit) classical shadow protocols~\cite{huang2020predicting, mao2025qudit} and (qudit-based) quantum learning protocols \cite{Huang2021,nöller2025infinitehierarchymulticopyquantum}. Another exciting direction is the study of higher-dimensional error-correcting codes, such as qudit stabilizer codes, where the tool can simulate code performance, track error propagation, and evaluate syndrome extraction under realistic noise conditions. The circuit fidelity concept also provides a natural foundation for average fidelity estimation \cite{kueng2016comparing,Burgholzer2021,Linden2021lightweight} and  randomized benchmarking \cite{PhysRevA.80.012304,magesan2011scalable}, allowing us to systematically assess gate and sequence performance across many shots. By supporting fast weak simulation and exact sampling of measurement outcomes, the framework can help explore noise-tailored compilation strategies, optimize qudit gate sequences, and generate quantitative predictions for experimentally relevant metrics. Altogether, QuickQudits provides a flexible platform for both theoretical studies and experimental investigations of multi-qudit entanglement, error mitigation, and the development of robust qudit quantum protocols.

\subsection*{Acknowledgements}
We would like to thank Jadwiga Wilkens for her valuable insights and feedback on the circuit fidelity estimation approach, as well as for inspiring many ideas for future work in randomized benchmarking. We also thank Paul Aigner and Alena Romanova for discussions about related tools and their research in this area. 
This work was funded by the Austrian Federal Ministry of Education,
Science and Research via the Austrian Research Promotion Agency (FFG) through the projects FO999914030
(MUSIQ) and FO999921407 (HDcode) and we thank all our colleagues from these consortia, whose support helped advance this research.
\end{multicols}

\printbibliography

\newpage

\section*{Appendix: Code examples}
\label{appendix}
This appendix collects compact usage examples that complement the algorithms and data structures introduced in the main text. We show (i) the construction and visualization of a small qudit Clifford circuit and (ii) the preparation and measurement of a qutrit GHZ state together with the corresponding stabilizer tableau output. The snippets are intentionally minimal and focus on the core API calls.
You can find many more examples on the official site \url{https://github.com/QUICK-JKU/QuickQudits}.

\begin{figure}[h]
\begin{tcolorbox}[
        enhanced, 
        width=0.4\textwidth, 
        nobeforeafter, 
        colback=gray!5, 
        equal height group=ex1,
        sharp corners,
        title=Circuit example
    ]
\begin{mintedbox}{python}
n, d = 4, 3

qc = QuantumCircuit(n, d, name="example_circuit")

qc.S(0)
qc.CX(0, 1)
qc.H(0)
qc.S(1)
qc.Hdag(2)
qc.Xdag(2)
qc.CZ(1, 3)
qc.SWAP(0, 2)
qc.CXdag(0, 2)
qc.CZ(2, 0)
qc.Hdag(3)
qc.Y(3)
qc.CXdag(3, 0)
qc.CZ(0, 1)
qc.measure_all()
qc.draw(show_info=True)
\end{mintedbox}
\end{tcolorbox}
\hfill 
    \begin{tcolorbox}[
        enhanced, 
        width=0.58\textwidth, 
        nobeforeafter, 
        colback=gray!5, 
        equal height group=ex1,
        valign=center, 
        halign=center, 
        sharp corners,
        title=Output
    ]
\begin{figure}[H]
\centering
\includegraphics[scale=0.2]{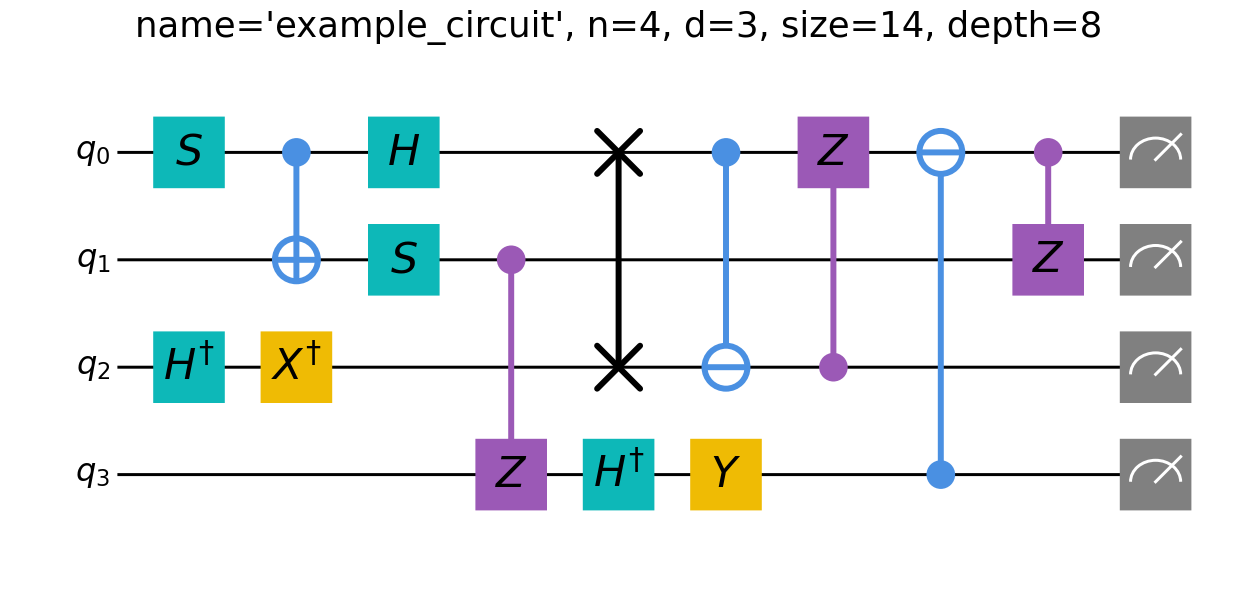}
\end{figure}
\end{tcolorbox}
\end{figure}

Note that we introduce a new symbol here, which is similar to the standard CNOT but with a minus instead of a plus on the target. This represents the daggered version of a qudit CNOT.

\begin{figure}[t]
\begin{tcolorbox}[
        enhanced, 
        width=0.49\textwidth, 
        nobeforeafter, 
        colback=gray!5, 
        equal height group=mytabs,
        sharp corners,
        title=GHZ measurement
    ]
\begin{mintedbox}{python}
n, d = 3, 3

qc_ghz = QuantumCircuit(n,d)

qc_ghz.H(0)
qc_ghz.CNOT(0,1)
qc_ghz.CNOT(1,2)

tab=Tableau(n, d, full=True)

qc_ghz.draw()

tab.apply_circuit(qc_ghz)
print("GHZ tableau")
print(tab, "\n")

mmt, res = tab.measure_all()
print("Post-measurement tableau")
\end{mintedbox}
    \end{tcolorbox}
    \hfill 
    \begin{tcolorbox}[
        enhanced, 
        width=0.49\textwidth, 
        nobeforeafter, 
        colback=gray!5, 
        equal height group=mytabs,
        valign=center, 
        halign=center, 
        sharp corners,
        title=Output
    ]
    \centering
        \includegraphics[scale=0.4]{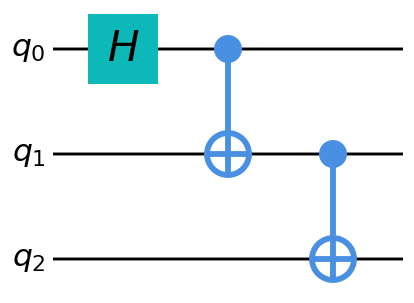}    
        \small
        \begin{Verbatim}[baselinestretch=0.9, fontsize=\small]
GHZ tableau 
#  | x0 x1 x2 | z0 z1 z2 | tau
------------------------
d0 |  0  0  0 |  1  0  0 | 0
d1 |  0  1  1 |  0  0  0 | 0
d2 |  0  0  1 |  0  0  0 | 0
------------------------
s0 |  2  2  2 |  0  0  0 | 0
s1 |  0  0  0 |  2  1  0 | 0
s2 |  0  0  0 |  0  2  1 | 0 

Post-measurement tableau
#  | x0 x1 x2 | z0 z1 z2 | tau
------------------------
d0 |  1  1  1 |  0  0  0 | 0
d1 |  0  1  1 |  0  0  0 | 0
d2 |  0  0  1 |  0  0  0 | 0
------------------------
s0 |  0  0  0 |  1  0  0 | 2
s1 |  0  0  0 |  2  1  0 | 0
s2 |  0  0  0 |  0  2  1 | 0 

Measurement result: [2 2 2]
        \end{Verbatim}
    \end{tcolorbox}
\end{figure}
\end{document}

\typeout{get arXiv to do 4 passes: Label(s) may have changed. Rerun}